\documentclass[10pt]{article}
\usepackage{palatcm, color, verbatim}

\usepackage[usenames,dvipsnames]{xcolor}

\usepackage{amsthm}
\usepackage{mdwmath}
\usepackage{amsmath}
\usepackage{amssymb}
\usepackage{cite}
\usepackage{epsfig}
\usepackage{url}
\usepackage{graphicx}
\usepackage{color}

\usepackage{tikz}
\usetikzlibrary{petri}
\usetikzlibrary{fit}
\usetikzlibrary{arrows,matrix,positioning}
\usepgflibrary{shapes.multipart} 

\newtheorem{rem}{Remark}
\newtheorem{theo}{Theorem}
\newtheorem{prop}{Proposition}
\newtheorem{defin}{Definition}
\newtheorem{lem}{Lemma}
\newtheorem{cor}{Corollary}
\newtheorem{intro_theo}{Theorem}

\newcommand\blfootnote[1]{%
  \begingroup
  \renewcommand\thefootnote{}\footnote{#1}%
  \addtocounter{footnote}{-1}%
  \endgroup
}

\begin{document}

\author{Dimitris S. Papailiopoulos and  Alexandros G. Dimakis\\
The University of Texas at Austin\\
\texttt{dimitris@utexas.edu}, \texttt{dimakis@austin.utexas.edu}
}

\title{
Locally Repairable Codes
}
\maketitle

\begin{abstract}

Distributed storage systems for large-scale applications typically use replication for reliability. 
Recently, erasure codes were used to reduce the large storage overhead, while increasing data reliability.
A main limitation of off-the-shelf erasure codes is their high-repair cost during single node failure events.
A major open problem in this area has been the design of codes that {\it i)} are repair efficient and {\it ii)} achieve arbitrarily high data rates. 

In this paper, we explore the repair metric of {\it locality}, which corresponds to the number of disk accesses required during a {\color{black}single} node repair.
Under this metric we characterize an information theoretic trade-off that binds together locality, code distance, and the storage capacity of each node.
We show the existence of optimal {\it locally repairable codes} (LRCs) that achieve this trade-off.
The achievability proof uses a locality aware  flow-graph gadget which leads to a randomized code construction.
Finally, we present an optimal and explicit LRC  that achieves arbitrarily high data-rates.
Our locality optimal construction is based on simple combinations of Reed-Solomon blocks.

\end{abstract}

\blfootnote{Parts of this work were presented in \cite{papailiopoulos2012locally}.}
\blfootnote{This research was supported in part by NSF Career Grant CCF-1055099 and research gifts by Intel, Microsoft Research, and Google Research.}
\blfootnote{A significant portion of this work has been completed while the authors were with the department of Electrical Engineering at the University of Southern California.}

\section{Introduction}
\label{sec:intro}

Traditional architectures for large-scale storage rely on systems that provide reliability through block replication.
The major disadvantage of replication is the large storage overhead.
As the amount of stored data is growing faster than hardware infrastructure, this becomes a major data center cost bottleneck. 
{\it Erasure coding} techniques achieve higher data reliability with considerably smaller storage overhead~\cite{weatherspoon2002erasure}. 
For that reason various erasure codes are currently implemented and deployed in production storage clusters. 
Applications where coding techniques are being currently deployed include cloud storage systems like Windows Azure~\cite{huang2012erasure}, big data analytics clusters 
({\it e.g.}, the Facebook Analytics Hadoop cluster~\cite{sathiamoorthy2013xoring}), archival storage systems, and peer-to-peer storage systems like Cleversafe and Wuala.

It is now well understood that classical erasure codes (such as Reed-Solomon {\color{black}codes}) are highly suboptimal for distributed storage settings~\cite{dimakis2010network}.
For example, the Facebook analytics Hadoop cluster discussed in~\cite{sathiamoorthy2013xoring}, deployed Reed-Solomon encoding for 8\% of the stored data. 
This 8\% of the stored data was reported to generate repair traffic that was approximately equal to 20\% of the total network traffic. 
The fact that traditional erasure codes are not optimized for {\color{black}node repairs}, is the main reason why they are not widely deployed in current storage systems. 

Three major repair cost metrics have been identified in the recent literature: 
{\it i)} the number of bits communicated in the network, also known as the {\it repair-bandwidth} \cite{dimakis2010network, rashmi2011optimal,suh2011exact, tamo2011mds, cadambe2011optimal, papailiopoulos2011repair}, 
{\it ii)} the number of bits read during each repair, {\color{black}{\it i.e.,}} the {\it disk-I/O} \cite{khan2011search,tamo2011mds}, and 
{\it iii)} more recently the number of nodes that participate in the repair process, also known as {\it repair locality}. 
Each of these metrics is more relevant for different systems and their fundamental limits are not completely understood. 

In this work, we focus on the metric of repair locality~\cite{gopalan2011locality,oggier2011self}. Consider a code of length $n$, with $k$ information symbols. 
A symbol $i$ has locality $r_i$, if it can 
be reconstructed by accessing $r_i$ other symbols in the code. 
For example, in an $(n,k)$ maximum-distance separable (MDS) code, 
every symbol has trivial locality $k$. We will say that a systematic code has \textit{information-symbol locality} $r$,
if all the $k$ information symbols have locality $r$.
Similarly, a code has \textit{all-symbol locality} $r$, if all $n$ symbols have locality $r$.

Different repair metrics optimize alternative objectives which may be useful in various storage systems 
depending on the specific architectures and workloads. 
Locality allows repairs by communicating with a very small subset of nodes. 
However, codes with {\color{black}small} locality are suboptimal in terms of the repair bandwidth and disk-I/O metrics.
 Further, as we show in this paper, 
LRCs must {\color{black}either sacrifice some code distance,} or use more storage compared to MDS codes to achieve {\color{black}low} locality. 
A recent alternative family of storage codes that seems to be practically applicable and offers higher storage efficiency and {\color{black}small repair bandwidth}
was proposed in~\cite{rashmi2013solution}.
One important benefit of codes with small locality is that their simple designs are easily implementable in distributed file systems like Hadoop~\cite{sathiamoorthy2013xoring} and Windows Azure Storage~\cite{huang2012erasure}. Further, codes with low locality were recently deployed in production clusters~\cite{huang2012erasure} and operating systems like Windows Server and Windows $8.1$~\cite{HuangSDC}.

Codes with {\color{black}small} locality were initially introduced in~\cite{han2007reliable,huang2007pyramid}.
Gopalan \textit{et al.}~\cite{gopalan2011locality} pioneered the theoretical study of locality by discovering a trade-off between code distance and information-symbol locality.
In~\cite{gopalan2011locality} the trade-off was obtained for scalar linear codes, {\it i.e.}, codes where each source and coded symbol is represented by a scalar over some finite field{\color{black}, and the each coded symbol is a linear function of the source symbols}.
{\color{black}Bounds on the code-distance for a given locality as well as code constructions were presented}  in parallel and subsequent works \cite{prakash2012optimal, kamath2012codes,rawat2012optimal,tamo2013optimal}.
Some works extend the designs and distance bounds to the case where repair bandwidth and locality are jointly optimized, under multiple local failures \cite{kamath2012codes,rawat2012optimal}, and under security constraints \cite{rawat2012optimal}.

{\bf Our Contributions:} We generalize the prior work of \cite{gopalan2011locality} and provide a distance bound that is universal{\color{black}: it holds for both linear and nonlinear codes, while it allows both scalar and vector code designs, where input and output symbols can have arbitrary sizes.}
{\color{black}We proceed to show that this information theoretic trade-off is achievable, when $r+1$ divides the length of the code $n$.
We conclude with presenting explicit constructions for codes with all-symbol locality}.
 We provide a formal definition of an LRC and then proceed with stating  {\color{black}our three contributions in more detail}.

\begin{defin}
An $(n,r,d,M,\alpha)$-LRC is a code that takes a file of size $M$ bits, encodes it in $n$ coded symbols of size $\alpha$ bits, and any of these $n$ coded symbols can 
be reconstructed by accessing and processing at most $r$ other symbols.
Moreover, the minimum-distance of the code is d, {\it i.e.}, the file of size $M$ can be reconstructed by accessing any $n-d+1$ of the $n$ coded symbols.\footnote{
{\color{black}In comparison, the definition of an information-symbol (or all-symbol) locality code in \cite{gopalan2011locality} assumes that
 the encoding is a linear mapping from $k$ input to $n$ output symbols. 
 Moreover, the input and  output symbols are  assumed to be of the same size, i.e., of the same number of bits.
Our definition is more general:  both linear and non-linear codes are allowed, and the size of the input and output symbols can be different.}}
\end{defin}

\noindent Our three contributions follow:
\\
\noindent{{\it 1) An information theoretic bound on code distance $d$}:}
We present a bound that binds together the code distance $d$, the locality $r$, and the size of each coded symbol $\alpha$ ({\it i.e.}, the storage capacity of each node).
The bound is information theoretic and covers all codes, linear or nonlinear, and reads as follows:
\begin{intro_theo}
\label{theo:distance}
An $(n,r,d,M,\alpha)$-LRC, as defined above, has distance $d$ that is bounded as 
\begin{align}
d&\le n-\left\lceil\frac{M}{\alpha}\right\rceil-\left\lceil\frac{M}{r\alpha}\right\rceil+2.\nonumber
\end{align}
\end{intro_theo}
\noindent
We establish our bound using an impossibility result for values of distance $d$ larger than the above.
The impossibility result uses an algorithmic proof similar to \cite{gopalan2011locality} and counting arguments on the entropy of subsets of coded symbols.
{\color{black}We would like to note that when we set $M=k$ and $\alpha =1$, which corresponds to the scalar-code regime, we obtain the same bound as \cite{gopalan2011locality}, that is 
\begin{align}
d&\le n-k-\left\lceil\frac{k}{r}\right\rceil+2.\nonumber
\end{align}}
\vspace{0.1cm}
\\
\noindent
{\it 2) Achievability of the distance bound  {\color{black}when} $(r+1)$ divides $n$:}
\begin{intro_theo}
Let $(r+1)|n$ and {\color{black}$r\le n-d$}. 
Then, there exist $(n,r,d,M,\alpha)$-LRCs with minimum code distance 
$$d=n-\left\lceil\frac{M}{\alpha}\right\rceil-\left\lceil\frac{M}{ra}\right\rceil+2,$$
over a sufficiently large finite field.
\label{theo:achieve}
\end{intro_theo}
\noindent 
We prove the achievability using a novel information flow-graph gadget, in a similar manner to \cite{dimakis2010network}.
In contrast to \cite{dimakis2010network}, the flow-graph that we construct is finite, locality aware, and simpler to analyze.
The existence of  $(n,r,d,M,\alpha)$-LRCs is established through a capacity achieving scheme on a multicast network \cite{ho2006random}, specified by the aforementioned flow-graph.
The obtained LRCs are vector codes: codes where each source and coded symbol is represented as a vector (not necessarily of the same length).
{\color{black}This is yet another case where vector codes are employed to achieve an optimal trade-off. In \cite{dimakis2010network}, the codes achieving the optimal repair bandwidth-storage trade-off are also vector linear.}
\vspace{0.1cm}
\\ 
\noindent{{\it 3) Explicit code constructions of optimal LRCs:}}
We construct explicit LRCs for the following set of coding parameters:
$$\left(n,r,d=n-k+1, M, \alpha = \frac{r+1}{r}\cdot \frac{M}{k}\right), \text{ such that } (r+1)|n.$$
Our codes are optimal when $(r+1)\nmid k$.
The above parameters correspond to codes with all-symbol locality $r$ and rate $\left(1+\frac{1}{r}\right)\cdot \frac{k}{n}$, where any $k$ coded symbols suffice to recover the file.
Our designs are vector-linear and each symbol stored requires only $r\cdot O( \log(n))$ bits in its representation.
We show that these codes not only have optimal locality, but also admit simple node repairs based on XORs.

The remainder of this paper is organized as follows.
In Section~\ref{sec:preliminaries}, we provide the coding theoretic definitions used in the {\color{black}subsequent} sections.
In Section~\ref{sec:distance_bound}, we provide a distance bound for codes with all-symbol locality.
In Section~\ref{sec:lrc_achievability}, we prove that this bound is achievable using {\color{black}random vector codes}.
In Section~\ref{sec:lrc_constructions}, we provide an explicit LRC construction, and discuss its properties.

\section{Preliminaries}
\label{sec:preliminaries}

A way to calculate the code distance of a linear code is through its generator matrix: calculating the minimum distance is equivalent to finding the largest set of columns of the generator matrix that are not full-rank \cite{gopalan2011locality, kamath2012codes}.
In the following, we use entropy to characterize the distance of a code.
This is the key difference to the related works in \cite{gopalan2011locality}, \cite{kamath2012codes}, which provide results only for linear codes.
The use of the entropy of coded symbols ensures that our bounds are universal: they hold for linear and nonlinear codes, for any file and coded symbol size, irrespective of a vector or scalar representation.
The main properties that we exploit here are the following: 
entropy is oblivious to the encoding process (linear or nonlinear), it can accommodate different input or output symbol sizes, and different  symbol representations (scalar or vector).
We will now proceed with our technical discussion.

Let a file of size $M$ bits\footnote{
The $M$ file elements can also be elements of an appropriate $q$-ary alphabet, for any $q\ge 2$. 
We keep the discussion in bits for simplicity.} 
be represented as an $M$-dimensional vector ${\bf x}$, whose elements can be considered as independent and identically distributed (i.i.d.) uniform random variables, each drawn from a Galois Field $\mathbb{GF}(2)$, referred to as $\mathbb{F}_2$ for convenience.\footnote{
We assume that ${\bf x}$ consists of  i.i.d. uniform random variables, since all $M$ bits are assumed to hold the same amount of useful information (are of equal entropy).}
The (binary) entropy of ${\bf x}$ will then be\footnote{If the base alphabet was $q$-ary instead of binary, then we would need to use $q$-ary entropies.}
\begin{equation}
H({\bf x}) =
 M.
\end{equation}
Moreover, 
let $G:\mathbb{F}_2^M\mapsto \mathbb{F}_2^{ n\cdot \alpha}$ be an encoding (generator) function, that takes as input the file ${\bf x}$ and maps it to $n$ coded symbols, each of size $\alpha$:
\begin{equation}
G({\bf x}) = {\bf y} = [Y_1\ldots Y_n]\nonumber
\end{equation}
where each encoded symbol has entropy
\begin{equation}
H(Y_i)\le\alpha,
\nonumber
\end{equation}
for all $i\in[n]$, where $[n]=\{1,\ldots, n\}$.
In the following, we {\color{black}frequently} refer to $\alpha$ as the storage cost per coded symbol.

The generator function $G$ defines an $n$-length code $\mathcal{C}$.
The {\color{black}effective data} rate of the code is the ratio of the total source entropy to the aggregate entropy of the stored encoded information
\begin{equation}
R = \frac{H({\bf x})}{\sum_{i=1}^n H(Y_i)}.\nonumber
\end{equation}
{\color{black}We continue with a definition for the minimum code distance.}
\begin{defin}[Minimum code distance]
The minimum distance $d$ of the code $\mathcal{C}$ is equal to the minimum number of erasures of {\color{black}coded} symbols in ${\bf y}$ after which the entropy of the non-erased symbols is strictly less than $M$, that is,
\begin{equation}
d = \min_{H\left(\{Y_1,\ldots,Y_n\}\backslash \mathcal{E}\right)<M}|\mathcal{E}| \nonumber
\end{equation}
where $\mathcal{E}\in 2^{\{Y_1,\ldots, Y_n\}}$ and $2^{\{Y_1,\ldots, Y_n\}}$ is the power set of the symbols in $\{Y_1,\ldots, Y_n\}$.
\end{defin}
\noindent In other words, when a code has minimum distance $d$, {\color{black}this means that} there is sufficient entropy after any $d-1$ erasures of coded symbols to reconstruct the file.
The above definition can be restated in its dual form:
{\color{black}the minimum distance $d$ of the code $\mathcal{C}$ is equal to the length of the code $n$, minus the maximum number of coded symbols in ${\bf y}$  that cannot reconstruct the file, that is, }
\begin{equation}
d =n-\max_{H(\mathcal{S})<M}|\mathcal{S}|\nonumber
\end{equation}
where  $\mathcal{S}\in 2^{\{Y_1,\ldots, Y_n\}}$.
\begin{rem}
Observe that the above distance definition applies to linear, or nonlinear codes, and to any length of input and output symbols.
\end{rem}

\noindent 
We continue with the definition of  repair locality.
\begin{defin}[Repair Locality]
A coded symbol $Y_i$,  $i\in[n]$, is said to have repair locality $r$, if there exists at least one set of coded symbols with indices in $\mathcal{R}(i)\subseteq [n]\backslash \{i\}$, call it $Y_{\mathcal{R}(i)}$, of cardinality $|\mathcal{R}(i)|=r$, and a function $g_i:\mathbb{F}_2^{r\cdot \alpha}\rightarrow\mathbb{F}_2^\alpha$, 
such that $Y_i$ can be expressed as a function of these $r$ coded symbols, {\it i.e.}, 
$Y_i = g_i (Y_{\mathcal{R}(i)})$.
\end{defin}

\section{A Universal bound between code distance, locality, and storage cost}
\label{sec:distance_bound}
In this section, we provide an information theoretic bound for locally repairable codes.
Specifically, we answer the question:
what is the maximum possible distance $d$ of a code that has locality $r$?
We provide a universal upper bound on the minimum distance of a code of length $n$, with all-symbol locality $r$, where each coded symbol has size $\alpha$.
We do so by an algorithmic proof, in a similar manner to~\cite{gopalan2011locality}.
Deriving such a distance bound reduces {\color{black}to lower bounding} the cardinality of the largest set $\mathcal{S}$ of coded symbols whose entropy is less than $M$.

In our proof, the only structural property that we use, is the fact that every {\color{black}coded} symbol has locality $r$.
Specifically, if a code $\mathcal{C}$ has locality $r$, then for each of its coded symbols, say $Y_i$, there exist at least one group of at most $r$ other coded symbols $Y_{\mathcal{R}(i)}$ that can reconstruct $Y_i$, for $i\in[n]$.
We define as 
$$\Gamma(i) = \{i,\mathcal{R}(i)\}$$
 a set of $r+1$ coded symbols that has the property
\begin{equation}
H(Y_{\Gamma(i)})=H(Y_i,Y_{\mathcal{R}(i)})=H(Y_{\mathcal{R}(i)}) \le r\alpha, \nonumber
\end{equation}
for all $i\in[n]$; the above comes due to the functional dependencies induced by locality.
We refer to such a set of coded symbols as an {\it $(r+1)$-group}.
The theorem and its proof follow.

\begin{theo}
\label{theo:distance}
An $(n,r,d,M,\alpha)$-LRC has minimum distance $d$ that is bounded as
\begin{align}
d&\le n-\left\lceil\frac{M}{\alpha}\right\rceil-\left\lceil\frac{M}{r\alpha}\right\rceil+2\nonumber.
\end{align}
\end{theo}

\begin{proof}
In this proof we use some of the algorithmic techniques that {\color{black}were introduced} in~\cite{gopalan2011locality}.
Our aim is to lower bound the cardinality of a set $\mathcal{S}$, consisting of the {\it maximum} number of coded symbols with entropy $H(\mathcal{S})$ {\it strictly less} than the filesize $M$.
This bound will be equivalent to an upper bound on the minimum code distance $d$, since 
\begin{equation}
d =n-\max_{
\begin{smallmatrix}
\mathcal{S}\subset\{Y_1,\ldots, Y_n\}\\
H(\mathcal{S})<M
\end{smallmatrix}
}|\mathcal{S}|.\nonumber
\end{equation}

To build such a maximally sized set described above, we need to collect as many symbols as possible that have as small joint entropy as possible.
Subsets of coded symbols that have many dependencies (small joint entropy) are preferred to subsets of the the same cardinality, but of larger joint entropy.
The only structural information about the code that we can exploit to introduce dependencies is that of repair locality: every repair group $Y_{\Gamma(i)}$ has joint entropy at most $r\cdot \alpha$, while an arbitrary set of $r+1$ symbols can have joint entropy up to $(r+1)\cdot \alpha$.

We build the set $\mathcal{S}$ in an algorithmic way through iterative steps.
The algorithm picks as many  $(r+1)$-groups as possible, until it exits.
The algorithm that builds the set follows in Fig.~\ref{fig:buildS}.
We proceed with analyzing the size and entropy of the sets that it can possibly construct.
{\color{black}The goal of our analysis is to lower bound the size of the set $\mathcal{S}_l$ that the algorithm can  possibly produce.
This will tell us that no matter how the code is constructed, its minimum distance cannot be more than $n-|\mathcal{S}_l|$.
}

\begin{figure}[h]
\begin{center}
\begin{tabular}{|c|l|}
\hline
step &\\
\hline
1 & Set $\mathcal{S}_0=\emptyset$ and $i=1$\\
2 & \texttt{WHILE} $H(\mathcal{S}_{i-1})<M$\\
3 & \hspace{0.5cm}Pick a coded symbol $Y_j\notin\mathcal{S}_{i-1}$\\
4 & \hspace{0.5cm}\texttt{IF} $H(\mathcal{S}_{i-1}\cup\{Y_{\Gamma(j)}\})<M$\\
5& \hspace{1cm}\texttt{set} $\mathcal{S}_i=\mathcal{S}_{i-1}\cup Y_{\Gamma(j)}$\\
6& \hspace{0.5cm}\texttt{ELSE IF} $H(\mathcal{S}_{i-1}\cup\{Y_{\Gamma(j)}\})\ge M$ \\
7& \hspace{1cm} $\mathcal{T} = \underset{\mathcal{T}'\subset \Gamma(j); H(Y_{\mathcal{T}'}\cup\mathcal{S}_{i-1})<M}{\arg\max} |\mathcal{T}'|$ \\
8& \hspace{1cm}\texttt{IF} $\mathcal{T} = \emptyset$ \\
9& \hspace{1.5cm} \texttt{EXIT} \\
10& \hspace{1cm}\texttt{ELSE}\\
11& \hspace{1.5cm} \texttt{set} $S_i = S_{i-1}\cup Y_\mathcal{T}$ \\
12& \hspace{1.5cm} \texttt{EXIT} \\
13& \hspace{0.5cm}$i=i+1$\\
\hline
\end{tabular}
\end{center}
\caption{The algorithm that builds set $\mathcal{S}$.}
\label{fig:buildS}
\end{figure}

We denote the collection of coded symbols at each step of the iteration as $\mathcal{S}_i$.
At each step $i$, the difference in cardinality between $\mathcal{S}_i$ and $\mathcal{S}_{i-1}$ is denoted as
\begin{equation}
s_i = |\mathcal{S}_i|-|\mathcal{S}_{i-1}|
\end{equation}
and the difference between the entropy of the two sets as
\begin{equation}
h_i = H(\mathcal{S}_i)-H(\mathcal{S}_{i-1}).
\end{equation}
The algorithm exits before reaching $H(\mathcal{S}_i)\ge M$.
There are two ways that the algorithm terminates: \\
{\it i)} it either collects $(r+1)$-groups until it exits at line  $9$, or\\ 
{\it ii)} the last subset of coded symbols that is added to $\mathcal{S}_{i-1}$ is smaller than $r+1$ and the algorithm exits at line $12$, after collecting some subset of an $(r+1)$-group, such that $H(\mathcal{S}_i)<M$ is not violated.\\
Let us denote by $l$ the {\it last} iteration of the algorithm during which a new {\it non-empty} set of coded symbols is added to the {\color{black}current} set of coded symbols.
We shall now proceed with lower bounding $|\mathcal{S}_l|$.
\\ 
\\
\underline{Case {\it i) The algorithm exits at line $9$:}}\\
Since the algorithm exits at $9$, this means that its last iteration is the $(l+1)$-st, where no more symbols are added.
Again, we denote by $l$ the {\it last} iteration during which our set of coded symbols is expanded by a non-empty set.
First observe that, for any $1\le i\le l$, we have
\begin{equation}
1 \le s_i \le r+1 \label{eq:si_bound}
\end{equation}
since at each iteration the algorithm augments the set $\mathcal{S}_{i-1}$ by at least one new symbol, {\it i.e.}, $Y_j$, 
which is always possible since $H(\mathcal{S}_{i-1})<M$, for all $i\le l$ and $H(Y_1,\ldots, Y_n)=M$.
Then, $s_i\le r+1$ is a consequence of the fact that 
$$|\mathcal{S}_{i}|=|\mathcal{S}_{i-1}\cup Y_{\Gamma(i)}|\le |\mathcal{S}_{i-1}|+|Y_{\Gamma(i)}|\le|\mathcal{S}_{i-1}|+r+1.$$
We also have that
\begin{equation}
h_i \le (s_i-1)\alpha.
\label{eq:hi_bound}
\end{equation}
To see why the above is true, let $\mathcal{S}_{i-1} = \mathcal{A}\cup \mathcal{B}$, where $\mathcal{B}=\mathcal{S}_{i-1}\cap Y_{\mathcal{R}(j)}$ is the subset of symbols from $Y_{\mathcal{R}(j)}$ that are already in $\mathcal{S}_{i-1}$ ($\mathcal{B}$ can be empty if no symbols from $\mathcal{R}(j)$ are in $\mathcal{S}_{i-1}$). 
Then, 
\begin{align}
H(\mathcal{S}_i)=H(\mathcal{S}_{i-1} \cup Y_{\Gamma(j)})&=H(\mathcal{S}_{i-1}\cup \{Y_{\mathcal{R}(j)}\backslash \mathcal{B}\}) \le H(\mathcal{S}_{i-1})+H(Y_{\mathcal{R}(j)}\backslash \mathcal{B}) \nonumber\\
&\le H(\mathcal{S}_{i-1})+|Y_{\mathcal{R}(j)}\backslash \mathcal{B}|\alpha \\
&= H(\mathcal{S}_{i-1})+(s_i-1)\alpha\nonumber,
\end{align}
where the second equality comes from the fact that $Y_j$ is a function of some symbols in $\mathcal{S}_{i-1}\cup \{Y_{\mathcal{R}(j)}\backslash \mathcal{B}\}$, due to locality, and the last equality is due to 
$$s_i=|\mathcal{S}_{i}|-|\mathcal{S}_{i-1}|= |Y_{\Gamma(j)}\backslash \mathcal{B}|=|Y_{\mathcal{R}(j)}\backslash \mathcal{B}|+1.$$
From \eqref{eq:hi_bound}, we also obtain
\begin{equation}
\alpha\cdot s_i\ge h_i+\alpha\label{eq:si_hi_bound}.
\end{equation}

Now, we can start bounding the size of  $\mathcal{S}_l$ as follows
\begin{align}
\alpha |\mathcal{S}_l|& = \alpha\sum_{i=1}^{l}s_i\overset{\eqref{eq:si_hi_bound}}{\ge} \sum_{i=1}^{l} \left(h_i+\alpha\right) = \left(\sum_{i=1}^{l} h_i\right)+l\cdot \alpha=H(\mathcal{S}_l)+l\cdot \alpha. \label{eq:Slbound1}
\end{align}
We continue with lower bounding the two quantities  in \eqref{eq:Slbound1}: $H(\mathcal{S}_l)$ and $l\cdot \alpha$.
First observe that since the algorithm is exiting, it means that the aggregate entropy $H(\mathcal{S}_l)=\sum_{i=1}^{l} h_i$ is so large that no other symbol can be added to our current set $\mathcal{S}_l$, without violating the entropy condition.
Hence,
\begin{equation}
H(\mathcal{S}_l)\ge M-\alpha. \label{eq:HSl_bound}
\end{equation}
Assume otherwise, {\it i.e.}, for example $H(\mathcal{S}_l)\le M-\alpha-\epsilon$, for any $\epsilon>0$.
Then, any coded symbol not in $\mathcal{S}_l$ can be added in $\mathcal{S}_l$ so that the aggregate entropy is at most $M-\epsilon$: the new symbol can only increase the joint entropy by at most $\alpha$. 
Hence, $H(\mathcal{S}_l)$ has to be at least $M-\alpha$.

Now we will lower bound $l$, the number of iterations to reach an entropy of at least $M-\alpha$.
Since the algorithm is assumed to exit at line 9, as mentioned before, at every iteration $i$ we have $s_i\le r+1$ and $h_i\le (s_i-1) \alpha$, for all $1\le i\le l$.
The minimum number of iterations occurs, when at each iteration the algorithm picks sets such that the entropy increase $h_i$  is equal to its upper bound $r\cdot \alpha$.
Therefore,
\begin{align}
l \ge \left\lceil \frac{H(S_l)}{r \cdot \alpha} \right\rceil\ge \left\lceil \frac{M-\alpha}{r \cdot \alpha} \right\rceil.\label{eq:l_bound}
\end{align}
Using \eqref{eq:HSl_bound} and \eqref{eq:l_bound}, we can rewrite \eqref{eq:Slbound1} as
\begin{align}
\alpha |\mathcal{S}_l|& \ge H(\mathcal{S}_l)+l\cdot \alpha\ge M-\alpha+ \alpha\cdot \left\lceil \frac{M-\alpha}{r \cdot \alpha} \right\rceil \nonumber\\
\Rightarrow |\mathcal{S}_l|& \ge \left\lceil \frac{M-\alpha+ \alpha\cdot \left\lceil \frac{M-\alpha}{r \cdot \alpha} \right\rceil}{\alpha} \right\rceil = 
\left\lceil \frac{M}{\alpha}-1+ \left\lceil \frac{M-\alpha}{r \cdot \alpha} \right\rceil \right\rceil\nonumber\\
&\overset{(i)}{=} \left\lceil \frac{M}{\alpha}\right\rceil-1+ \left\lceil \frac{M-\alpha}{r \cdot \alpha} \right\rceil = \left\lceil \frac{M}{\alpha}\right\rceil-1+ \left\lceil \frac{M}{r \cdot \alpha}-\frac{1}{r} \right\rceil\nonumber\\
&\ge \left\lceil \frac{M}{\alpha}\right\rceil-1+ \left\lceil \frac{M}{r \cdot \alpha} \right\rceil-1\nonumber\\
\Rightarrow d&\le n-|S_l|\le n -  \left\lceil \frac{M}{\alpha}\right\rceil - \left\lceil \frac{M}{r \cdot \alpha} \right\rceil+2, \label{eq:d_bound_1}
\end{align}
where the equality in $(i)$ comes from the fact that $\lceil x+n\rceil=\lceil x \rceil+n$, for any real number $x$ and any integer $n$ \cite{knuth1989concrete}; in our case $x = \frac{M}{\alpha}$ and $n = -1+ \left\lceil \frac{M-\alpha}{r \cdot \alpha}\right\rceil$.\\

\noindent
\underline{Case {\it ii) The algorithm exits at line $12$:}}\\
In this case, the algorithm runs for $l$ iterations; during the $l-1$ first iterations, the algorithm augments $\mathcal{S}_{i-1}$ at step $1\le i\le l-1$, by entire $(r+1)$-groups. 
During the last step $i=l$, the algorithm augments $\mathcal{S}_{l-1}$, by a {\it partial} subset of $Y_{\Gamma(j)}$, for some coded symbol $Y_j$ not in $\mathcal{S}_{l-1}$.
From the above, we get the following bounds
\begin{equation}
s_i \le r+1, \text{ for all }1\le i\le l-1, \text{ and }s_l \le r,\label{eq:si_bound_2}
\end{equation}
and 
\begin{equation}
h_i \le (s_i-1)\alpha, \text{ for all }1\le i\le l-1, \text{ and }h_l \le s_l\alpha\label{eq:hi_bound_2}.
\end{equation}
The right most part of the above bounds comes from the fact that, during the last iteration, at most $r$ coded symbols can be added to the set $\mathcal{S}_{l-1}$.
Moreover, in contrast to all other iterations, during the last iteration it is possible to augment $\mathcal{S}_{l-1}$ by $s_l$ new coded symbols, all being independent to each other and any other symbols in $\mathcal{S}_{l-1}$; that is $h_l$ can be as large as $s_l\alpha$.

We will again bound the size of $\mathcal{S}_l$, the maximal set of coded symbols that has entropy less than $M$.
We use \eqref{eq:hi_bound_2} and sum over all $s_i$'s to obtain our bound on the size of $\mathcal{S}_l$:
\begin{align}
&\sum_{i=1}^l h_i\overset{\eqref{eq:hi_bound_2}}{\le}\left(\sum_{i=1}^{l-1}s_i\alpha \right)-(l-1)\cdot \alpha+s_l\cdot \alpha\nonumber\\
\Rightarrow &\alpha\cdot \sum_{i=1}^{l}s_i\ge \sum_{i=1}^l h_i+(l-1)\cdot \alpha\nonumber
\Rightarrow \alpha\cdot |\mathcal{S}_l|\ge \sum_{i=1}^l h_i+(l-1)\cdot \alpha \nonumber \\
\Rightarrow &\alpha\cdot |\mathcal{S}_l|\ge H(\mathcal{S}_l)+(l-1)\cdot \alpha. \label{eq:Sl_bound_2}
\end{align}

We now need to bound again the two quantities that control the bound in \eqref{eq:Sl_bound_2}: $H(\mathcal{S}_l)$ and $l$.
We can use the same bound as used in Case {\it i)} for $H(\mathcal{S}_l)=\sum_{i=1}^{l} h_i$, {\it i.e.}, the entropy of the constructed set $H(\mathcal{S}_l)$ has to be large enough, so that another iteration cannot be carried on:
\begin{equation}
H(\mathcal{S}_l)\ge M-\alpha. \label{eq:HSl_bound_2}
\end{equation}
Again, let us assume otherwise: $H(\mathcal{S}_l)\le M-\alpha-\epsilon$, for any $\epsilon>0$.
Then, any coded symbol not in $\mathcal{S}_l$ can be added in $\mathcal{S}_l$ so that the aggregate entropy is at most $M-\epsilon$.
Hence, $H(\mathcal{S}_l)\ge M-\alpha$.

Now we will bound the number of iterations $l$.
The last added subset of symbols $\mathcal{T}$ that augments $\mathcal{S}_{l-1}$ has cardinality less than, or equal to $r$.
{\color{black}Otherwise, if $\mathcal{T}$} was an entire $(r+1)$-group, then the statement in line $6$ of the algorithm would have been FALSE.
This means that adding $(r+1)$-groups for all iterations, including the $l$-th one, can have as much entropy as $r\cdot l\cdot \alpha$, which has to be {\it at least} as much as $M$, or else we would not have been under {\it Case ii)} of the algorithm.
Hence,
\begin{equation}
r\cdot l\cdot \alpha \ge M \Rightarrow l\ge \left\lceil \frac{M}{r\cdot \alpha}\right\rceil. \label{eq:l_bound_2}
\end{equation}

We can now use the bounds in \eqref{eq:HSl_bound_2} and \eqref{eq:l_bound_2} to rewrite \eqref{eq:Sl_bound_2} as
\begin{align}
&\alpha |\mathcal{S}_l| \ge H(\mathcal{S}_l)+(l-1)\cdot \alpha\ge M-\alpha+ \alpha\cdot \left(\left\lceil \frac{M}{r \cdot \alpha} \right\rceil-1\right)\nonumber\\
\Rightarrow &|\mathcal{S}_l| \ge \left\lceil \frac{M-2\alpha+ \alpha\cdot \left\lceil \frac{M}{r \cdot \alpha} \right\rceil}{\alpha} \right\rceil =
\left\lceil \frac{M}{\alpha}-2+ \left\lceil \frac{M}{r \cdot \alpha} \right\rceil \right\rceil\overset{(i)}{=} \left\lceil \frac{M}{\alpha}\right\rceil-2+ \left\lceil \frac{M}{r \cdot \alpha} \right\rceil \nonumber\\
\Rightarrow &d\le n-|S_l|\le n -  \left\lceil \frac{M}{\alpha}\right\rceil - \left\lceil \frac{M}{r \cdot \alpha} \right\rceil+2,\label{eq:d_bound_2}
\end{align}
{\color{black}where $(i)$ comes from the fact that $\lceil x+n\rceil=\lceil x \rceil+n$, for any real number $x$ and any integer $n$ \cite{knuth1989concrete}.}
The bounds of \eqref{eq:d_bound_1} and \eqref{eq:d_bound_2} establish our theorem.
\end{proof}


\begin{rem}
Observe that when $(r+1)|n$, we can partition the set of $n$ coded symbols in $\frac{n}{r+1}$ non-overlapping $(r+1)$-groups.
The algorithmic proof that we used, relied on the fact that collecting $(r+1)$-groups, is one of the ways to achieve the lower bound on the size of $\mathcal{S}$.
This observation will lead us to an achievability proof for the case of $(r+1)|n$.
We will see that pair-wise disjoint repair groups is one of the (possibly many) arrangements of repair groups that leads to optimal constructions.
\end{rem}
\begin{rem}
In the above bound, if we set $\alpha = 1$ and $M=k$, we get the same bound as~\cite{gopalan2011locality}.
The $\alpha = 1$ case is equivalent to considering scalar codes.
As it turns out, for the scalar case, linear codes are sufficient for this bound and nonlinearity in the encoding process does not come with any improvements in code distance.
\end{rem}



In the following section, we show that the above distance bound is tight when \textsc{$(r+1)|n$}.
{\color{black}This does not rule out that the bound is tight under more general assumptions, however, this is left as an open question.
For linear codes, \cite{gopalan2011locality} shows that codes with information-symbol locality can be constructed under different assumptions  (for example when $r|k$ and $2<d<r+3$), using a structure theorem ({\it e.g.,} see Theorem 15 in \cite{gopalan2011locality}).
At the same time, it is impossible to construct optimal and linear LRCs (with all-symbol locality) when $2<d<r+3$ and $r|k$ ({\it e.g.,} see Corollary 10 in \cite{gopalan2011locality}). 
It would be interesting to explore the use of the tools presented in \cite{gopalan2011locality}, to provide further impossibility, or achievability results that extend the $(r+1)|n$ case that we study.
}

\section{Achievability of the Bound: Random LRCs}
\label{sec:lrc_achievability}

In this section, we {\color{black}establish} the following existence result:
\begin{theo}
Let $(r+1)|n$ and {\color{black}$r\le n-d$}. 
Then, there exist $(n,r,d,M,\alpha)$-LRCs with minimum distance 
$$d=n-\left\lceil\frac{M}{\alpha}\right\rceil-\left\lceil\frac{M}{ra}\right\rceil+2,$$
over a sufficiently large finite field.
\label{theo:achieve}
\end{theo}

We establish the theorem through capacity achieving {\color{black}schemes} for a specific communication network; such network will be defined through a directed and acyclic flow-graph.
In the first subsection, we introduce the communication model for our network.
In the second subsection, we show that a capacity achieving scheme for the aforementioned network maps to specific codes with specified parameters.
In the third and fourth subsections, 
we construct randomized capacity achieving schemes and then map them to $(n,r,d,M,\alpha)$-LRCs.
The distance $d$ of the aforementioned codes will be equal to the upper bound of Theorem 1.

\subsection{The flow-graph network, multicast sessions, and its multicast capacity}

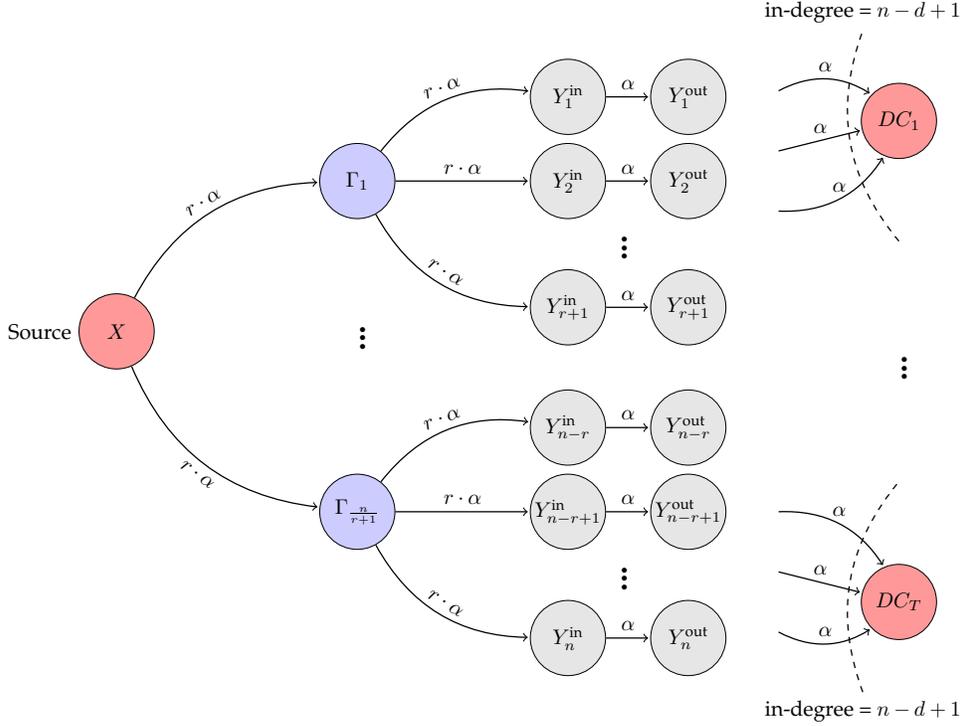
\begin{figure*}[t]
\hrulefill
\begin{center}
\scalebox{0.8}{
\begin{tikzpicture}
[every place/.style={minimum size=12.5mm, minimum width=12.5mm}]
\node[ minimum size=55mm, place,draw=black,  fill=red!40, label = left:{Source}] (source) at (-3,-2.5) [place] {$X$} ;
\node[ minimum width=15mm, place,draw=black,  fill=blue!20] (gamma1in) at (1,0) [place] {$\Gamma_1$} ;
\node[ minimum width=15mm, place,draw=black,  fill=black!10] (Y1in) at (3.5+1,1.3+0.1) [place] {$Y_1^{\text{in}}$} ;
\node[ minimum width=15mm, place,draw=black,  fill=black!10] (Y1out) at (5.5+1,1.3+0.1) [place] {$Y_1^{\text{out}}$} ;

\node[ minimum width=15mm, place,draw=black,  fill=black!10] (Y2in) at (3.5+1,0) [place] {$Y_2^{\text{in}}$} ;
\node[ minimum width=15mm, place,draw=black,  fill=black!10] (Y2out) at (5.5+1,0) [place] {$Y_2^{\text{out}}$} ;

\node[ minimum width=15mm, place,draw=black,  fill=black!10] (Yrin) at (3.5+1,-2-0.1) [place] {$Y_{r+1}^{\text{in}}$} ;
\node[ minimum width=15mm, place,draw=black,  fill=black!10] (Yrout) at (5.5+1,-2-0.1) [place] {$Y_{r+1}^{\text{out}}$} ;
\node[text width=4cm] at (7.35,-1) {\huge $\vdots$};
\path (Y1in) edge[semithick, post] node [above, midway]  {$\alpha$} (Y1out);
\path (Y2in) edge[semithick, post] node [above, midway]  {$\alpha$} (Y2out);
\path (Yrin) edge[semithick, post] node [above, midway]  {$\alpha$} (Yrout);
\path (gamma1in) edge[semithick, post, bend left] node [midway, sloped, above] {$r\cdot\alpha$} (Y1in);
\path (gamma1in) edge[semithick, post] node [midway, above] {$r\cdot\alpha$} (Y2in);
\path (gamma1in) edge[semithick, post, bend right] node [midway, sloped, above] {$r\cdot\alpha$} (Yrin);
\node[text width=4cm] at (3,-2.5) {\huge $\vdots$};
\node[ minimum width=15mm, place,draw=black,  fill=blue!20] (gammarin) at (1,0-5.5) [place] {$\Gamma_{\frac{n}{r+1}}$} ;
\node[ minimum width=15mm, place,draw=black,  fill=black!10] (Yr1in) at (3.5+1,1.3+0.1-5.5) [place] {$Y_{n-r}^{\text{in}}$} ;
\node[ minimum width=15mm, place,draw=black,  fill=black!10] (Yr1out) at (5.5+1,1.3+0.1-5.5) [place] {$Y_{n-r}^{\text{out}}$} ;

\node[ minimum width=15mm, place,draw=black,  fill=black!10] (Yr2in) at (3.5+1,0-5.5) [place] {$Y_{n-r+1}^{\text{in}}$} ;
\node[ minimum width=15mm, place,draw=black,  fill=black!10] (Yr2out) at (5.5+1,0-5.5) [place] {$Y_{n-r+1}^{\text{out}}$} ;

\node[ minimum width=15mm, place,draw=black,  fill=black!10] (Yrrin) at (3.5+1,-2-0.1-5.5) [place] {$Y_n^{\text{in}}$} ;
\node[ minimum width=15mm, place,draw=black,  fill=black!10] (Yrrout) at (5.5+1,-2-0.1-5.5) [place] {$Y_n^{\text{out}}$} ;

\node[text width=4cm] at (7.35,-6.5) {\huge $\vdots$};
\node[ minimum width=15mm, place,draw=black,  fill=red!40] (DC1) at (9+1,1) [place] {$DC_1$} ;
\node[ minimum width=15mm, place,draw=black,  fill=red!40] (DC4) at (9+1,-7) [place] {$DC_T$} ;
\path  (7+1,1.5) edge[semithick, post, bend left] node [above, midway]  {$\alpha$} (DC1);
\path  (7+1,0.5) edge[semithick, post] node [above, midway]  {$\alpha$} (DC1);
\path  (7+1,-0.5) edge[semithick, post, bend right] node [above, midway]  {$\alpha$} (DC1);
\path  (9+1,-1) edge[semithick, dashed, bend left] node [above, at end]  {in-degree = $n-d+1$ } (8.4+1, 2.5	);

\node[text width=4cm] at (12,-3) {\huge $\vdots$};

\path  (7+1,1.5-7) edge[semithick, post, bend left] node [above, midway]  {$\alpha$} (DC4);
\path  (7+1,0.5-7) edge[semithick, post] node [above, midway]  {$\alpha$} (DC4);
\path  (7+1,-0.5-7) edge[semithick, post, bend right] node [above, midway]  {$\alpha$} (DC4);
\path  (8.4+1,-1-7-0.5) edge[semithick, dashed, bend left] node [below, at start]  {in-degree = $n-d+1$ } (9+1, 2.5-7-0.5);

\path (Yr1in) edge[semithick, post] node [above, midway]  {$\alpha$} (Yr1out);
\path (Yr2in) edge[semithick, post] node [above, midway]  {$\alpha$} (Yr2out);
\path (Yrrin) edge[semithick, post] node [above, midway]  {$\alpha$} (Yrrout);
\path (gammarin) edge[semithick, post, bend left] node [midway, sloped, above] {$r\cdot\alpha$} (Yr1in);
\path (gammarin) edge[semithick, post] node [midway, above] {$r\cdot\alpha$} (Yr2in);
\path (gammarin) edge[semithick, post, bend right] node [midway, sloped, above] {$r\cdot\alpha$} (Yrrin);
\path (source) edge[semithick, post, bend left] node [midway, sloped, above] {$r\cdot \alpha$} (gamma1in);
\path (source) edge[semithick, post, bend right] node [midway, sloped, below] {$r\cdot \alpha$} (gammarin);
\end{tikzpicture}
}
\end{center}

\caption{ 
The directed acyclic information flow-graph $\mathcal{G}(n,r,d,\alpha)$. 
The left-most vertex is the source node $X$. 
{\color{black}The $\frac{n}{r+1}$  vertices $\Gamma_i$ correspond to nodes that limit the in-flow to specific groups of intermediate nodes.
The right-most $T={n \choose n-d+1}$  vertices $\text{DC}_i$ are the destination nodes (referred to as the data collectors) of the network.}
Each DC is connected to a different  $(n-d+1)$-tuple of $Y_{i}^{\text{out}}$ nodes.
} 
\label{fig:flow_graph}
\hrulefill
\end{figure*}

Our achievability proof relies on using random linear network coding (RLNC) on a directed acyclic flow-graph, borrowing ideas 
from~\cite{ho2006random,dimakis2010network}. Fig.~\ref{fig:flow_graph} shows the directed and acyclic flow-graph that we use, which is formally defined subsequently.

{\color{black} At a conceptual level our proof analyzes a nested multicast problem that consists of two parts. 
We show that when RLNC is  employed on the flow-graph in Fig.~\ref{fig:flow_graph}, {\it i)} it multicasts the source transmitted by node $X$ to all data collectors (global decoding requirements) and {\it ii)} it simultaneously allows each collection of $r$ nodes $Y^{\text{out}}_j$, originating from the same $\Gamma_i$ node, to reconstruct whatever $\Gamma_i$ transmits (local decoding requirements).
The first 
part of the proof is a standard application of RLNC~\cite{ho2006random}.
For the second part, our proof relies on a further subtle technicality that we discuss below.

General nested multicasting problems can be very challenging, but our problem has a very special structure:
there are no edges between $Y_j^{\text{in}}$, $Y_{j'}^{\text{out}}$ nodes that originate from different $\Gamma_i$ vertices.
This means that there is no ``algebraic interference" between the linear combinations of packets transmitted/received by these nodes.
We use this fact to show that if the $T$ data collectors $\text{DC}_1,\ldots, \text{DC}_T$ receive linearly independent equations of the source information, then 
 each group of $r$ nodes $Y^{\text{out}}_j$ that originate from the same $\Gamma_i$ node, receive linearly independent equations of the packets that node $\Gamma_i$ transmits.
This allows us to essentially use the technique of Ho \textit{et al.}~\cite{ho2006random} to establish that both the global and local decoding requirements are simultaneously satisfied. }

We now proceed with the detailed description of our proof. Let $\mathcal{G}(n,r,d,\alpha)$, be a directed acyclic graph that represents a communication network with $1$ source node and $T$ destination nodes and has vertex set
\begin{align}
\mathcal{V} = 
&\left\{X,
\Gamma_1,
\ldots, 
\Gamma_{\frac{n}{r+1}},
Y_1^{\text{in}},Y_1^{\text{out}},
\ldots,
Y_n^{\text{in}},Y_n^{\text{out}},
\text{DC}_1,\ldots, \text{DC}_T
\right\},\nonumber
\end{align}
where $X$ denotes the {\it source node}, $\text{DC}_1,\text{DC}_2,\ldots, \text{DC}_T$ are the $T={n \choose n-d+1}$ {\it destination nodes}, referred to as the Data Collectors (DCs), and the remaining nodes are the {\it intermediate nodes}.
Each vertex in $\mathcal{V}$ is assumed to be a receive and/or transmit {\it node}.
It will become clear what that means after the following definitions, which are introduced to make our proof self-contained, while requiring minimal familiarity with network coding theory.
For further details on our network model please refer to~\cite{yeung2005network}.

\begin{defin}[edge capacity/network use/local encoding function/source message]
A directed edge between two vertices $v$ and $u$ denotes a communication link between two nodes, over which bits are transmitted.
All links are assumed to introduce no error.
The directed edge capacity $c(v,u)$, between vertices $v,u$, denotes the maximum number of bits that can be communicated from node $v$ to node $u$ {\color{black}during} a single network use. 
A single network use denotes the sequence of single transmissions over every directed edge.
A message ${\bf m}_{(v,u)}$ is a collection of $c(v,u)$ bits that are transmitted from node $v$ to node $u$, during a single network use.
A message ${\bf m}_{(v,u)}$ can be considered as a collection of $c(v,u)$ binary uniform variables\footnote{Although we assume that the messages transmitted over the links are sets of binary variables, this can be generalized to $2^\tau$-ary variables ({\it i.e.}, each variable will now be an element of a finite field of order $q=2^\tau$).
This is possible, if we consider $\tau$ consecutive transmit sessions per link, during a network use.
We can then consider an equivalent network where the alphabet is $2^\tau$-ary.
As a consequence, the entropies used under this setting should be $2^\tau$-ary, and all the following results carry on to that case.}
with joint binary entropy equal to 
$H({\bf m}_{(v,u)}) = c(v,u).$
Let $\mathcal{I}_v$ denote the (in-coming) vertices incident to vertex $v$.
Then, the message ${\bf m}_{(v,u)}$ that is transmitted on a link $e(v,u)$ during a single network use, is the output of a {\it local encoding function}
$$f_{v,u}\left(\left\{{\bf m}_{(v',v)};v'\in\mathcal{I}_v\right\}\right):\mathbb{F}_2^{\sum_{v'\in\mathcal{I}_v}c(v',v)}\rightarrow \mathbb{F}_2^{c(v,u)},$$ 
that takes as input the set of messages $\left\{{\bf m}_{(v',v)};v'\in\mathcal{I}_v\right\}$ received by node $v$ via the incident nodes in $\mathcal{I}_v$.
The source node of the network holds a source bit sequence ${\bf x}$ of size $H({\bf x})$ bits and wishes to transmit it to the $T$ destination nodes. 
\end{defin}

We are now ready to define the directed weighted edge (link) set, that is determined by the following link capacities
\begin{equation}
c(v,u)=\left\{
\begin{array}{cl}
r\cdot \alpha,& (v,u) = (X,\Gamma_j), \forall j\in\left[\frac{n}{r+1}\right],\\
r\cdot \alpha ,& (v,u) = (\Gamma_j, Y^{\text{in}}_l), \forall j\in\left[\frac{n}{r+1}\right] \text{ and } l\in \{(j-1)(r+1)+1,\ldots,j(r+1)\},\\
\alpha, & (v,u) = \left(Y_j^{\text{in}},Y_j^{\text{out}} \right),\forall j\in[n],\\
\alpha ,& (v,u) = (Y_j^{\text{out}}, \text{DC}_t), \forall j\in \mathcal{F}_t \text{ and } t\in[T],\\
0, &\text{otherwise},
\end{array}
\right.\nonumber
\end{equation}
where the $T$ sets $\mathcal{F}_1,\ldots, \mathcal{F}_T$ are all  $T={n \choose n-d+1}$ possible subsets of $n-d+1$ integers in $[n]$.
Observe that the in-degree of any $\text{DC}_t$ node (the number of incident directed edges arriving at {\color{black}that} node) is $n-d+1$.

The $\mathcal{G}(n,r,d,\alpha)$ network comes together with $T$ {\it decoding requirements}:
each destination node $\text{DC}_t$, for $t\in[T]$, is required, after a network use, to be able to reproduce
from its received bits the source sequence ${\bf x}$.
{\color{black}The decoding requirements can be stated as the following conditional entropy requirements:}
$$\mathcal{D}_t:\;\;H\left({\bf x}\left|\left\{{\bf m}_{(Y_j^{\text{out}},DC_t)}: j\in\mathcal{F}_t\right\}\right.\right)=0, \;\forall t\in[T].$$
We are now ready to provide the main definition needed for our proof.
\begin{defin}[multicast capacity and capacity achieving schemes]
The directed graph $\mathcal{G}(n,r,d,\alpha)$
 and a set of decoding requirements $\mathcal{D}_1,\ldots, \mathcal{D}_T$,  specify a {\it multicast connection problem}.
Let $C$ be the maximum number of source bits such that all decoding requirements are satisfied, after a single network use.
Then, $C$ is defined as the multicast capacity of $\mathcal{G}(n,r,d,\alpha)$.
A capacity achieving scheme, is a collection of local encoding functions such that all $T$ decoding requirements are satisfied for $H({\bf x})=C$.
\end{defin}

In the following subsection, we derive a connection between a capacity achieving scheme on $\mathcal{G}(n,r,d,\alpha)$ and the existence of a code of well specified parameters.
Then, we calculate the capacity of $\mathcal{G}(n,r,d,\alpha)$.

\subsection{Connecting capacity achieving schemes to codes}

{\color{black}The following lemma  connects capacity achieving schemes on $\mathcal{G}(n,r,d,\alpha)$, to  the existence of codes.}
\begin{lem}
The set of $n$ local encoding functions $f_{(Y^{\text{in}}_i,Y^{\text{out}}_i)}$, $i\in[n]$, of a capacity achieving scheme on $\mathcal{G}(n,r,d,\alpha)$, can be mapped to a code of length $n$, that encodes a file of size $C$ in $n$ coded symbols, each of size $\alpha$ bits.
This code has distance $d$.
\end{lem}
\begin{proof}
Observe that any local encoding function $f_{v,u}$ can be re-written as some {\it global encoding function} of the $C$ source bits in ${\bf x}$ \cite{yeung2005network}.
Let $f_i({\bf x}):\mathbb{F}_2^M\rightarrow \mathbb{F}_2^\alpha$ be the global function representation for $f_{(Y^{\text{in}}_i,Y^{\text{out}}_i)}$.
If the $n$ local encoding functions $f_{(Y^{\text{in}}_1,Y^{\text{out}}_1)}$,\ldots,$f_{(Y^{\text{in}}_n,Y^{\text{out}}_n)}$ are part of a capacity achieving scheme, then due to the decoding requirements being satisfied, we have
\begin{equation}
H\left(\left.{\bf x}\right|\left\{{\bf m}_{(Y_j^{\text{out}},\text{DC}_t)}: j\in\mathcal{F}_t\right\}\right)=0 \Rightarrow H\left({\bf x} \left| \left\{f_j({\bf x}): j\in\mathcal{F}_t\right\}\right.\right)=0
\end{equation}
since ${\bf m}_{(Y_j^{\text{out}},\text{DC}_t)}$ is a function of $f_j({\bf x})$, $t\in[T]$.
Now, let 
$Y_ i = f_i({\bf x})$
and observe that each $Y_i$ is a collection of $\alpha$ bits.
Then, the $T$ decoding requirements $H\left({\bf x} \left| \left\{f_i({\bf x}): i\in\mathcal{F}_t\right\}\right.\right)=0$, for $t\in[T]$, are equivalent to the following statement: 
``any collection of $n-d+1$ symbols $Y_i$, with $i\in[n]$, are sufficient to reconstruct ${\bf x}$".
This implies that 
\begin{equation}
[Y_1,\ldots Y_n ] = \left[f_1({\bf x})\ldots f_n({\bf x})\right],
\end{equation}
defines a code of length $n$, that encodes a files of size $C$, each coded symbol is of size $\alpha$, and any $n-d+1$ coded symbols can reconstruct ${\bf x}$, {\it i.e.},  the code has distance $d$.
\end{proof}

\begin{rem}
Observe that the above result does not guarantee that the code defined by $ f_{(Y^{\text{in}}_i,Y^{\text{out}}_i)}$ has locality $r$.
Locality comes as an artifact of the graph structure and the random capacity achieving scheme that we will use.
\end{rem}

\subsection{Computing the source-destination cuts and achieving the capacity}
In this subsection, we calculate the capacity of $\mathcal{G}(n,r,d,\alpha)$, and show how to achieve it.
Let us first define the minimum cuts of the above network.

\begin{defin}[minimum cut]
A directed cut between nodes $v$ and $u$, referred to as $\text{Cut}(v,u)\subseteq E$, is a subset of directed edges, such that if these edges are removed, then there does not exist a directed path between nodes $v$ and $u$;
$|\text{Cut}(v,u)|$ is the sum of all edge capacities in the set $\text{Cut}(v,u)$, referred to as the capacity of $\text{Cut}(v,u)$.
A minimum cut $\text{MinCut}(v,u)$ is the cut with the minimum aggregate edge capacity among all cuts between $v$ and $u$.
\end{defin}
It is a well-known fact for communication networks, that $|\text{MinCut}(v,u)|$ is an upper bound on the number of bits that one can communicate from node $v$ to node $u$ \cite{yeung2005network}.
Consequently, the cut with the minimum capacity, among the cuts of all source-destination pairs, is an {\it upper bound} on the multicast capacity of a network.
Most importantly, we know that this bound is achievable for multicast session networks.
We state as Theorem 3, what is a collection of breakthrough results from \cite{ahlswede2000network,ho2006random}.
\begin{theo}[\cite{ahlswede2000network,ho2006random} ]
The multicast capacity $C$ of a network with $1$ source and $T$ destination nodes, is equal to the minimum number among all capacities of minimum source-destination cuts.
The capacity is achievable using random linear network coding.
\end{theo}

\begin{rem}
In our case, RLNC stands for having local encoding functions $f_{v,u}\left(\left\{{\bf m}_{(v',v)};v'\in\mathcal{I}_v\right\}\right):\mathbb{F}_q^{\sum_{v'\in\mathcal{I}_v}c(v',v)}\rightarrow \mathbb{F}_q^{c(v,u)}$, for all $u,v$, such that the outputs of each of those functions are $c(v,u)$ symbols over a $q$-ary alphabet, and each output symbol is a linear combination of the $\sum_{v'\in\mathcal{I}_v}c(v',v)$ input symbols; each of these linear combinations has coefficients that are picked uniformly at random from the $q$-ary alphabet. 
\end{rem}

We use the above results and definitions to prove the key technical lemma of this subsection.
Before we proceed with that, we present some properties of the ceiling and floor functions that are used in our proof.
\begin{prop}[\cite{knuth1989concrete} ]
\label{prop:ceil_floor}
Let $n$ and $m$ be positive integer numbers, and $x$ any real number.
Then, the following hold\\
\text{(i)} $\left\lfloor\frac{n}{m}\right\rfloor =\left\lceil\frac{n+1}{m}\right\rceil-1$,
\text{(ii)}  $\left\lceil \frac{x+m}{n}\right\rceil =\left\lceil \frac{\lceil x \rceil+m}{n}\right\rceil$,
\text{(iii)}  $\left\lceil \frac{\lceil x/m\rceil}{n}\right\rceil = \left\lceil \frac{x}{nm}\right\rceil$.
\end{prop}

\noindent We now proceed with the main lemma.
\begin{lem}
Let  $d = n-\left\lceil\frac{M}{\alpha}\right\rceil-\left\lceil\frac{M}{r\cdot \alpha}\right\rceil+2.$
Then, then the multicast capacity of the $\mathcal{G}(n,r,d,\alpha)$ network is equal to
\begin{equation}
C= \left\lceil\frac{M}{\alpha}\right\rceil\alpha\ge M.
\end{equation}
\end{lem}
\begin{proof}
Using Theorem 3, the capacity of $\mathcal{G}(n,r,d,\alpha)$ is equal to 
$$\min_{t\in[T]} \left|\text{MinCut}\left(X,\text{DC}_t\right)\right|.$$
Let us calculate the minimum cut capacity among all minimum cuts.
Let us denote as the $i$-th $(r+1)$-group, the set of $r+1$ intermediate nodes $Y_j^{\text{out}}$ that can be reached from $\Gamma_i$.
Now consider a DC that connects to a set of $n-d+1$ nodes {\it including} all the nodes of, say, the first $(r+1)$-group, and assume without loss of generality that this is $\text{DC}_1$.
There are two (meaningful) choices for $\text{Cut}(X,\text{DC}_1)$:
{\color{black}{\it i)} it can consist of all $(r+1)$ edges $(Y_i^{\text{in}},Y_i^{\text{out}})$, $i\in [r+1]$, of the $(r+1)$-group,
 or
 {\it ii)} it can consist of simply the $(X,\Gamma_1)$ edge.}\footnote{{\color{black}The  assumption $r\le n-d$ is made such that $n-d+1\ge r+1$. This implies that a DC has to connect to at least $r+1$ nodes.}}
 
Clearly, the latter choice leads to a smaller cut capacity, since  $(X,\Gamma_1)$  has capacity $r\cdot \alpha$, whereas the $r+1$ edges $(Y_i^{\text{in}},Y_i^{\text{out}})$, $i\in [r+1]$, have an aggregate capacity of $(r+1)\cdot \alpha$.
Hence, for every cut that includes $r+1$ edges of the $(Y_j^{\text{in}},Y_j^{\text{out}})$ kind that belong to the same $(r+1)$-group, say the $i$-th, then $(X,\Gamma_i)$ can be {\color{black}used} instead, reducing the capacity of such cut.
Therefore, the smallest source-destination cut is the one that contains the largest possible number of $(X,\Gamma_i)$ edges.

Now, the minimum aggregate capacity among all these $T$ cuts, {\it i.e.}, $ \min_{t=1,\ldots, T} |\text{MinCut}(X,\text{DC}_t)|$, will be the one that corresponds to the minimum cut of the DC that covers entirely as many  $(r+1)$-groups as possible.
Since the total number of $Y_j^{\text{out}}$ nodes that a DC connects to is $n-d+1$, then the number of $(r+1)$-groups it can entirely cover is\footnote{We would like to note here that the ratio inside the floor function is never an integer number: if it was, then all DCs could connect to exactly one less $Y_i^{\text{out}}$ node while maintaining exactly the same source-destination cut.}
$\left\lfloor\frac{n-d+1}{r+1}\right\rfloor.$
The minimum cut will hence include 
$$n_1 = \left\lfloor\frac{n-d+1}{r+1}\right\rfloor$$ edges of the $(X,\Gamma_i)$ kind, which contribute to the cut an aggregate capacity of
$n_1 r\alpha$.
The remaining capacity comes from cutting a number of 
$$n_2 = n-d+1-n_1=n-d+1-(r+1)\left\lfloor\frac{n-d+1}{r+1}\right\rfloor$$ edges of the $(Y_i^{\text{in}},Y_i^{\text{out}})$ kind.
Therefore, we have that the smallest source-DC cut is equal to 
{\small
\begin{align}
\min_{t\in[T]} \left|\text{MinCut}\left(X,\text{DC}_t\right)\right|& = n_1\cdot r\cdot \alpha + n_2\cdot \alpha
= \left(n-d+1-\left\lfloor\frac{n-d+1}{r+1}\right\rfloor\right)\alpha 
= \left(\left\lceil\frac{M}{\alpha}\right\rceil+\left\lceil\frac{M}{r\alpha}\right\rceil-1-\left\lfloor\frac{\left\lceil\frac{M}{\alpha}\right\rceil+\left\lceil\frac{M}{r\alpha}\right\rceil-1}{r+1}\right\rfloor\right)\alpha \nonumber\\
& \overset{(i)}{=} \left(\left\lceil\frac{M}{\alpha}\right\rceil+\left\lceil\frac{M}{r\alpha}\right\rceil-1-\left\lceil\frac{\left\lceil\frac{M}{\alpha}\right\rceil+\left\lceil\frac{M}{r\alpha}\right\rceil}{r+1}\right\rceil+1\right)\alpha\nonumber
\overset{(iii)}{=} \left(\left\lceil\frac{M}{\alpha}\right\rceil+\left\lceil\frac{M}{r\alpha}\right\rceil-\left\lceil\frac{\left\lceil\frac{M}{\alpha}\right\rceil+\left\lceil\frac{\lceil M/\alpha \rceil}{r}\right\rceil}{r+1}\right\rceil\right)\alpha\\
&
\overset{(ii)}{=} \left(\left\lceil\frac{M}{\alpha}\right\rceil+\left\lceil\frac{M}{r\alpha}\right\rceil-\left\lceil\frac{\left\lceil\frac{M}{\alpha}\right\rceil+\left\lceil\frac{M}{\alpha}\right\rceil\frac{1}{r}}{r+1}\right\rceil\right)\alpha\nonumber= \left(\left\lceil\frac{M}{\alpha}\right\rceil+\left\lceil\frac{M}{r\alpha}\right\rceil-\left\lceil\frac{\left\lceil\frac{M}{\alpha}\right\rceil\frac{r+1}{r}}{r+1}\right\rceil\right)\alpha\nonumber\\
&= \left(\left\lceil\frac{M}{\alpha}\right\rceil+\left\lceil\frac{M}{r\alpha}\right\rceil-\left\lceil\frac{\left\lceil\frac{M}{\alpha}\right\rceil}{r}\right\rceil\right)\alpha\overset{(iii)}{=} \left(\left\lceil\frac{M}{\alpha}\right\rceil+\left\lceil\frac{M}{r\alpha}\right\rceil-\left\lceil\frac{M}{r\alpha}\right\rceil\right)\alpha =\left\lceil\frac{M}{\alpha}\right\rceil\cdot\alpha\ge M,
\end{align}
}where on the second, third, and fourth lines of derivations, we explicitly state which of the three properties of the ceiling/floor function {\color{black}found in Proposition~\ref{prop:ceil_floor}} we are using.
The above establishes our lemma.
\end{proof}

\noindent Using Lemma 2, Theorem 3, and Lemma 1, we obtain the following corollary.
\begin{cor}
There exists a capacity achieving scheme for $\mathcal{G}(n,r,d,\alpha)$, whose local encoding functions $f_{(Y^{\text{in}}_i,Y^{\text{out}}_i)}$, for $i\in[n]$, map to a code of length $n$ that encodes $M^*\in\left[M, \left\lceil \frac{M}{\alpha}\right\rceil\alpha\right]$ source symbols in $n$ coded symbols of size $\alpha$, and the code has distance\footnote{Observe that here we obtain the result for $M^*\in \left[M, \left\lceil \frac{M}{\alpha}\right\rceil\alpha\right]$.
One can easily inspect that by substituting the $M^*$ value in the distance bound, this value does indeed respect it.}
$d = n-\left\lceil\frac{M}{\alpha}\right\rceil-\left\lceil\frac{M}{r\cdot \alpha}\right\rceil+2$.
\end{cor}

Observe that we are not done yet: we still have to prove that be above code has locality $r$.
The next subsection finalizes our proof, by showing that RLNC on $\mathcal{G}(n,r,d,\alpha)$ indeed implies codes with locality $r$ and distance matching our bounds, for a sufficiently large finite field.

\subsection{Establishing the locality of the code and concluding the proof}

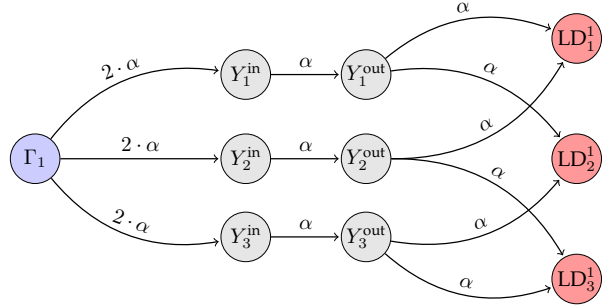
\begin{figure}[h]
\hrulefill
\begin{center}
\scalebox{0.8}{
\begin{tikzpicture}
\node[ minimum width=15mm, place,draw=black,  fill=blue!20] (gamma1in) at (1,0) [place] {$\Gamma_1$} ;
\node[ minimum width=15mm, place,draw=black,  fill=black!10] (Y1in) at (3.5+1,1.3+0.1) [place] {$Y_1^{\text{in}}$} ;
\node[ minimum width=15mm, place,draw=black,  fill=black!10] (Y1out) at (5.5+1,1.3+0.1) [place] {$Y_1^{\text{out}}$} ;
\node[ minimum width=15mm, place,draw=black,  fill=black!10] (Y2in) at (3.5+1,0) [place] {$Y_2^{\text{in}}$} ;
\node[ minimum width=15mm, place,draw=black,  fill=black!10] (Y2out) at (5.5+1,0) [place] {$Y_2^{\text{out}}$} ;
\node[ minimum width=15mm, place,draw=black,  fill=black!10] (Yrin) at (3.5+1,-1.3) [place] {$Y_3^{\text{in}}$} ;
\node[ minimum width=15mm, place,draw=black,  fill=black!10] (Yrout) at (5.5+1,-1.3) [place] {$Y_3^{\text{out}}$} ;
\node[ minimum width=15mm, place,draw=black,  fill=red!40] (LD1) at (10,2) [place] {$\text{LD}_1^1$} ;
\node[ minimum width=15mm, place,draw=black,  fill=red!40] (LD2) at (10,0) [place] {$\text{LD}_2^1$} ;
\node[ minimum width=15mm, place,draw=black,  fill=red!40] (LD3) at (10,-2) [place] {$\text{LD}_3^1$} ;
\path (Y1in) edge[semithick, post] node [above, midway]  {$\alpha$} (Y1out);
\path (Y2in) edge[semithick, post] node [above, midway]  {$\alpha$} (Y2out);
\path (Yrin) edge[semithick, post] node [above, midway]  {$\alpha$} (Yrout);
\path (gamma1in) edge[semithick, post, bend left] node [midway, sloped, above] {$2\cdot\alpha$} (Y1in);
\path (gamma1in) edge[semithick, post] node [midway, above] {$2\cdot \alpha$} (Y2in);
\path (gamma1in) edge[semithick, post, bend right] node [midway, sloped, above] {$2\cdot \alpha$} (Yrin);
\path (Y1out) edge[semithick, post, bend left] node [above, sloped, midway]  {$\alpha$} (LD1);
\path (Y2out) edge[semithick, post, bend right] node [above, sloped, midway]  {$\alpha$} (LD1);
\path (Y1out) edge[semithick, post, bend left] node [above, sloped, midway]  {$\alpha$} (LD2);
\path (Yrout) edge[semithick, post, bend right] node [above, sloped, midway]  {$\alpha$} (LD2);
\path (Y2out) edge[semithick, post, bend left] node [above, sloped, midway]  {$\alpha$} (LD3);
\path (Yrout) edge[semithick, post, bend right] node [above, sloped, midway]  {$\alpha$} (LD3);
\end{tikzpicture} 
}
\end{center}
\caption{The $\mathcal{G}_1$ subgraph induced by the first $(r+1)$-group of a $\mathcal{G}(n,r=2,d,\alpha)$ network.
{\color{black}The additional $\text{LD}^i_j$ are local data collectors that are conceptual.
These local DCs are not present in the original graph, and are used here to finalize the proof of Theorem \ref{theo:achieve}. 
We use them to establish the locality of the code obtained through the RLNC capacity achieving scheme on $\mathcal{G}(n,r=2,d,\alpha)$.}
}
\hrulefill
\label{fig:local_DC}
\end{figure}

To establish the locality of the code obtained in the previous subsection, we will show that an extra set of {\it local decoding requirements} are satisfied when RLNC is used.
For this part of the proof we will focus on the subgraphs induced by the $(r+1)$-groups.
Let $\mathcal{G}_{i}$, be the subgraph that is induced by the vertices 
$$\mathcal{V}_i = \left\{\Gamma_i, 
Y^{\text{in}}_{(i-1)\cdot(r+1)+1}, 
Y^{\text{out}}_{(i-1)\cdot(r+1)+1},
\ldots,
Y^{\text{in}}_{i\cdot(r+1)}, \ldots, Y^{\text{out}}_{i\cdot(r+1)}\right\},$$
for any $i\in \left[\frac{n}{r+1}\right]$.
Let us assume that for each of these subgraphs there exists an additional number of ${r+1\choose r}=r+1$ {\it local Data Collector} nodes, $\text{LD}^i_1,\ldots,\text{LD}^i_{r+1}$.
Each local DC is connected to one of the $r+1$ possible $r$-subsets of $Y_j^{\text{out}}$ nodes of $\mathcal{G}_i$, with 
$$j \in \{(i-1)(r+1)+1,\ldots,  i(r+1)\}.$$
In Fig. \ref{fig:local_DC}, we give an example of $\mathcal{G}_1$ with the added local DCs.

{\color{black}Each of these local DCs has a decoding requirement: it requires to be able to decode what was transmitted by $\Gamma_i$.}
Let us refer to such a decoding requirement for the $j$-th local DC of $\mathcal{G}_i$ as $\mathcal{LD}^i_j$.
{\color{black}
\begin{rem}
\label{rem:locality_via_LDC}
Observe that the decoding requirement $\mathcal{LD}^i_j$ implies that the $j$-th local DC can reconstruct any single of the $r+1$ messages ${\bf m}_{(Y^{\text{in}}_j,Y^{\text{out}}_j)}$, with
$j \in \{(i-1)(r+1)+1,\ldots, i(r+1)\}.$ 
This is true since all these $r+1$ messages are functions of what is transmitted by the $\Gamma_i$ node.
\end{rem}
The above observation will be used to establish the locality of the codes obtained from RLNC on  $\mathcal{G}(n,r,d,\alpha)$. }
Before we do that, we will state the following lemma, which will help us to conclude our proof.

\begin{lem}
When RLNC is used on $\mathcal{G}_i$, the decoding requirement $\mathcal{LD}^i_j$ is equivalent to a full-rank requirement $\mathcal{FR}_j^i$ on an $r\cdot \alpha\times r\cdot \alpha$ matrix with random i.i.d. coefficients.
\end{lem}
\begin{proof}
Without loss of generality, let us consider $\mathcal{G}_1$, moreover, let for simplicity 
$${\bf z}_1\in \mathbb{F}_q^{1\times r\cdot \alpha}$$
 be the source message that $\Gamma_1$ wishes to transmit to the local DCs.
Since the capacity of a $(\Gamma_1,Y_j^{\text{in}})$ edge is $r\cdot \alpha$,  for $j\in[r+1]$, then node $Y_j^{\text{in}}$ receives ${\bf z}_1$.
Moreover, due to the RLNC scheme used, the coefficients of the random linear combinations in the local encoding functions are picked independently.
Hence, node $Y_j^{\text{in}}$ will transmit to node $Y_j^{\text{out}}$ a vector of $\alpha$ symbols: 
$${\bf z}_1{\bf A}_{1,j}$$
where ${\bf A}_{1,j}$ is an $r\cdot \alpha \times \alpha$ matrix of random i.i.d. coefficients.
Then, any node $Y_j^{\text{out}}$ transmits to the local DCs of $\mathcal{G}_1$ exactly what it received, {\it i.e., } ${\bf z}_1{\bf A}_{1,j}$, since the capacity of the edge $(Y_j^{\text{out}}, \text{LD}_i^1)$ is $\alpha$.
Hence, any local DC receives $r$ vectors of size $\alpha$, which if put together form a vector of size $r\cdot \alpha$; this vector, for local DC $j$, can be re-written as ${\bf z}_1{\bf A}_j^1$, where ${\bf A}_j^1$ is an $r\cdot \alpha \times r\cdot \alpha$ matrix of random i.i.d. coefficients.
{\color{black}Hence, any local DC decoding requirement is equivalent a requirement on a square matrix of random coefficients being full-rank.
Let us refer to this full-rank requirement as $\mathcal{FR}_j^i$.}
\end{proof}

{\color{black}Observe that $\mathcal{FR}_j^i$ is a requirement that can be stated independently of the existence of local DCs.
Hence, we can now go back on $\mathcal{G}(n,r,d,\alpha)$ and show that RLNC allows all local decoding requirements and all $\mathcal{FR}_j^i$ conditions to be satisfied {\it simultaneously}.}

\begin{lem}
Let 
$$d = n-\left\lceil\frac{M}{\alpha}\right\rceil-\left\lceil\frac{M}{ra}\right\rceil+2$$
and let us employ RLNC on $\mathcal{G}(n,r,d,\alpha)$.
Then, all decoding requirements $\mathcal{D}_i$ of $\mathcal{G}(n,r,d,\alpha)$ and  all full rank requirements $\mathcal{FR}_j^i$ will be simultaneously satisfied, with nonzero probability, when the finite field is sufficiently large.
\end{lem}
\begin{proof}
Let $E_{\mathcal{G}}$ denote the event that some of the $T$ DCs of $\mathcal{G}(n,r,d,\alpha)$ cannot decode ${\bf x}$ successfully, which say, has probability $p_1$ that is a function of the size of the finite field used by the RLNC scheme \cite{ho2006random}.
Moreover, let $E_{\mathcal{FR}_j^i}$ denote the event that $\mathcal{FR}_j^i$ is not satisfied, which say, has probability $p_2$ that is also a function of the size of the finite field used by the RLNC scheme.
Then, the probability that RLNC does not satisfy some of the above conditions is
$$\Pr\left\{E_{\mathcal{G}} \bigcup \left\{\bigcup_{i=1}^{n/(r+1)} \bigcup_{j=(i-1)(r+1)+1}^{i(r+1)} E_{\mathcal{FR}_j^i} \right\}\right\}\le \Pr\left\{E_{\mathcal{G}}\right\}
+\sum_{i=1}^{n/(r+1)} \sum_{j=(i-1)(r+1)+1}^{i(r+1)}\Pr\left\{E_{\mathcal{FR}_j^i} \right\} = p_1+n\cdot p_2.$$
We can now conclude our proof, since $p_1$ and $p_2$ can be made arbitrarily small, using a sufficiently large finite field \cite{ho2006random}.
\end{proof}

Due to the above lemma and Lemma 1, we use the $f_1({\bf x}),\ldots, f_n({\bf x})$ global encoding functions (the global representations of the $f_{(Y_i^{\text{in}},Y_i^{\text{out}})}$s) of the RLNC scheme to obtain a code that encodes a file of size $M$ in $n$ coded symbols, each of size $\alpha$; such code has distance $d$.

Moreover, since all $\mathcal{FR}_j^i$ requirements are satisfied, then
{\color{black}as mentioned in Remark~\ref{rem:locality_via_LDC}}, each output of a global encoding function $f_i({\bf x})$ can be reconstructed from the outputs of a subset of $r$ other global encoding functions: this implies repair locality $r$.
Hence, the code defined by the global encoding functions $f_i$ is an $(n,r,d,M,\alpha)$-LRC, with 
 $$d = n-\left\lceil\frac{M}{\alpha}\right\rceil-\left\lceil\frac{M}{ra}\right\rceil+2.$$
This concludes the proof of Theorem 2.

\section{Locally Repairable Codes: Explicit Constructions}
\label{sec:lrc_constructions}

In this section, we provide an explicit LRC family for the operational point on the distance trade-off, where any $k$ subsets of coded nodes can reconstruct all $k$ file symbols, {\it i.e.}, when $d=n-k+1$.
For this regime, that resembles that of an $(n,k)$-MDS code, we will show how to achieve the distance of an $(n,k)$-MDS code,  while having locality $r<<n$ and {\color{black}sacrificing} only a small fraction of the code rate: the rate of our codes will be $\frac{1}{r}\frac{k}{n}$ less than that of an $(n,k)$-MDS code.
Specifically, the code parameters for our LRCs are 
$$\left(n,r,d=n-k+1, M, \alpha = \frac{r+1}{r}\cdot \frac{M}{k}\right), \text{ such that } (r+1)|n.$$
Our codes meet the optimal distance bound for all of the above coding parameters when $(r+1) \nmid k$.

The presented  codes come with the following {\color{black}design advantages}: 
{\it i)} they achieve arbitrarily high data rates,
{\it ii)} they can be constructed using Reed-Solomon encoded blocks,
{\it iii)} the repair of a lost node requires downloading blocks and XORing them at a destination node,
and {\it iv)} their vector size, or sub-packetization length, is $r$, {\color{black}and each stored sub-symbol is over a small finite field with size proportional to $n$.
This means that we can represent each coded symbol by using only $r\cdot O(\log n)$ bits.}



%

\begin{figure}[h]
\hrulefill
\vspace{-0.3cm}
\begin{center}
\scalebox{0.8}{
\begin{tikzpicture}
[node distance =0.1 cm and 5.5cm, every place/.style={rectangle, thick,minimum size=30mm, minimum width=20mm}, bend angle=23]
\node[ minimum width=15mm ,minimum size=5mm,place, draw=white, fill= white] (Y1out) at (2,2.5) [place] {\Large MDS Pre-coding and XORing} ;

\node[ minimum width=15mm, place,draw=white,  fill=white] (x1) at (-1.45,0) [place] {
$
\begin{array}{c}
x^{(1)}_1\;-\\
x^{(1)}_2\;-\\
x^{(1)}_3\;-\\
x^{(1)}_4\;-
\end{array}
$
} ;
\node[ minimum width=15mm, place,draw=white,  fill=white] (y1) at (1.45,-4) [place] {
$
\begin{array}{c}
-\;y^{(2)}_1\\
-\;y^{(2)}_2\\
-\;y^{(2)}_3\\
-\;y^{(2)}_4\\
-\;y^{(2)}_5\\
-\;y^{(2)}_6
\end{array}
$
} ;
\node[ minimum width=15mm, place,draw=white,  fill=white] (x2) at (-1.45,-4) [place] {
$
\begin{array}{c}
x^{(2)}_1\;-\\
x^{(2)}_2\;-\\
x^{(2)}_3\;-\\
x^{(2)}_4\;-
\end{array}
$
} ;
\node[ minimum width=15mm, place,draw=white,  fill=white] (y1) at (1.45,0) [place] {
$
\begin{array}{c}
-\;y^{(1)}_1\\
-\;y^{(1)}_2\\
-\;y^{(1)}_3\\
-\;y^{(1)}_4\\
-\;y^{(1)}_5\\
-\;y^{(1)}_6
\end{array}
$
} ;
\node[ minimum width=15mm, place,draw=white,  fill=white] (x1) at (-1.45,0) [place] {
$
\begin{array}{c}
x^{(1)}_1\;-\\
x^{(1)}_2\;-\\
x^{(1)}_3\;-\\
x^{(1)}_4\;-
\end{array}
$
} ;
\node[ minimum width=15mm, place,draw=white,  fill=white] (y1) at (5,-2) [place] {
$
\begin{array}{c}
s_1=y^{(1)}_1+y^{(2)}_1\\
s_2=y^{(1)}_2+y^{(2)}_2\\
s_3=y^{(1)}_3+y^{(2)}_3\\
s_4=y^{(1)}_4+y^{(2)}_4\\
s_5=y^{(1)}_5+y^{(2)}_5\\
s_6=y^{(1)}_6+y^{(2)}_6\\
\end{array}
$
} ;
\node[ minimum width=15mm, place,draw=black,  fill=blue!10] (Y1out) at (0,0) [place] {$(6,4)$-MDS} ;
\node[ minimum width=15mm, place,draw=black,  fill=red!40] (Y1out) at (0,-4) [place] {$(6,4)$-MDS} ;

\path (1.9,1.1) edge[semithick, post]  (3.8,-0.8);
\path (1.9,1.1-0.5) edge[semithick, post]  (3.8,-0.8-0.5);
\path (1.9,1.1-2*0.5) edge[semithick, post]  (3.8,-0.8-2*0.5);
\path (1.9,1.1-3*0.5) edge[semithick, post]  (3.8,-0.8-3*0.5);
\path (1.9,1.1-4*0.5) edge[semithick, post]  (3.8,-0.8-4*0.5);
\path (1.9,1.1-5*0.5) edge[semithick, post]  (3.8,-0.8-5*0.5);

\path (1.9,-2.8-0*0.5) edge[semithick, post]  (3.8,-0.8-0*0.5);
\path (1.9,-2.8-1*0.5) edge[semithick, post]  (3.8,-0.8-1*0.5);
\path (1.9,-2.8-2*0.5) edge[semithick, post]  (3.8,-0.8-2*0.5);
\path (1.9,-2.8-3*0.5) edge[semithick, post]  (3.8,-0.8-3*0.5);
\path (1.9,-2.8-4*0.5) edge[semithick, post]  (3.8,-0.8-4*0.5);
\path (1.9,-2.8-5*0.5) edge[semithick, post]  (3.8,-0.8-5*0.5);
\end{tikzpicture}}
\end{center}
\label{fig:code}
\end{figure}
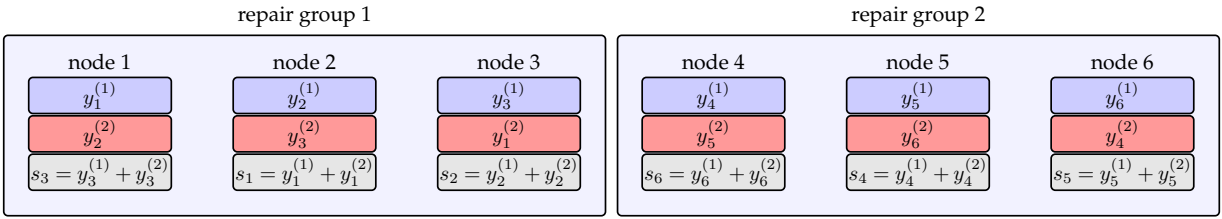
\begin{figure}[h]
\hrulefill
\begin{center}
\scalebox{0.8}{
\begin{tikzpicture}
[node distance =0.01 cm and 1cm, every place/.style={rectangle,  thick,minimum size=6mm, minimum width=23.7mm}, bend angle=23]
\node[place, rounded corners=2pt,  draw=white, fill=white] (node11) at (8.5,2) [place] {\Large Block Placement};

\node[place, rounded corners=2pt,  minimum size=30mm, minimum width =100mm, draw=black, fill=blue!5, label = above:{repair group 1}] (group1) at (3.4,-0.5) {};
\node[place, rounded corners=2pt,  minimum size=30mm, minimum width =100mm, draw=black, fill=blue!5, label = above:{repair group 2}] (group2) at(13.6,-0.5) {};

\node[place, rounded corners=2pt,  draw=black, fill=blue!20, label = above:{node 1}] (node11) {$y_1^{(1)}$};
\node[place, rounded corners=2pt,  draw=black, fill=red!40, below = of node11] (node12)  {$y_2^{(2)}$};
\node[place, rounded corners=2pt,  draw=black, fill=black!10, below = of node12] (node13)  {$s_3=y_3^{(1)}+y_3^{(2)}$};

\node[place, rounded corners=2pt,  draw=black, fill=blue!20, right = of node11, label = above:{node 2}] (node21) {$y_2^{(1)}$};
\node[place, rounded corners=2pt,  draw=black, fill=red!40, below = of node21] (node22)  {$y_3^{(2)}$};
\node[place, rounded corners=2pt,  draw=black, fill=black!10, below = of node22] (node23)  {$s_1 = y_1^{(1)}+y_1^{(2)}$};

\node[place, rounded corners=2pt,  draw=black, fill=blue!20, right = of node21, label = above:{node 3}] (node31) {$y_3^{(1)}$};
\node[place, rounded corners=2pt,  draw=black, fill=red!40, below = of node31] (node32)  {$y_1^{(2)}$};
\node[place, rounded corners=2pt,  draw=black, fill=black!10, below = of node32] (node33)  {$s_2 = y_2^{(1)}+y_2^{(2)}$};

\node[place, rounded corners=2pt,  draw=black, fill=blue!20, right = of node31, label = above:{node 4}] (node41) {$y_4^{(1)}$};
\node[place, rounded corners=2pt,  draw=black, fill=red!40, below = of node41] (node42)  {$y_5^{(2)}$};
\node[place, rounded corners=2pt,  draw=black, fill=black!10, below = of node42] (node43)  {$s_6=y_6^{(1)}+y_6^{(2)}$};

\node[place, rounded corners=2pt,  draw=black, fill=blue!20, right = of node41, label = above:{node 5}] (node51) {$y_5^{(1)}$};
\node[place, rounded corners=2pt,  draw=black, fill=red!40, below = of node51] (node52)  {$y_6^{(2)}$};
\node[place, rounded corners=2pt,  draw=black, fill=black!10, below = of node52] (node53)  {$s_4=y_4^{(1)}+y_4^{(2)}$};

\node[place, rounded corners=2pt,  draw=black, fill=blue!20, right = of node51, label = above:{node 6}] (node61) {$y_6^{(1)}$};
\node[place, rounded corners=2pt,  draw=black, fill=red!40, below = of node61] (node62)  {$y_4^{(2)}$};
\node[place, rounded corners=2pt,  draw=black, fill=black!10, below = of node62] (node63)  {$s_5=y_5^{(1)}+y_5^{(2)}$};
\end{tikzpicture}}
\end{center}
\caption{MDS pre-coding, XORing, and block placement in nodes.}
\hrulefill
\label{fig:code}
\end{figure}

\subsection{Code construction}

Let a file ${\bf x}$ of size $M = rk$ symbols\footnote{{\color{black}here the size of each symbol depends on the code construction, and is not necessarily binary. As we see in the following, the size of each symbol will be proportional to $\log(n)$ bits.}}, that is sub-packetized in $r$ parts,
 \begin{equation}
{\bf x} = \left[
{\bf x}^{(1)}
\ldots
{\bf x}^{(r)}
\right], \nonumber
 \end{equation}
with each ${\bf x}^{(i)}$, $i\in[r]$, having size $k$.
We encode each of the $r$ file parts independently, into coded vectors ${\bf y}^{(i)}$ of length $n$, where $(r+1)|n$, using an outer $(n,k)$ MDS code
 \begin{align}
{\bf y}^{(1)} = {\bf x}^{(1)}{\bf G},\;\;\ldots,\;\; {\bf y}^{(r)} = {\bf x}^{(r)}{\bf G},\nonumber
 \end{align}
where ${\bf G}$ is an $n\times k$ MDS generator matrix.

As MDS pre-codes, we use $(n,k)$-RS codes that require each of the $k$ elements to be over a finite field $\mathbb{F}_{2^p}$, for any $p$ such that $2^p\ge n$.
{\color{black}This will imply that all stored sub-symbols in our code are over a finite field of size $2^p\ge n$.}
We then generate a single parity XOR vector from all the coded vectors
 \begin{equation}
{\bf s} = \bigoplus_{i=1}^{r} {\bf y}^{(i)}. \nonumber
 \end{equation}

The above {\it pre-coding} process yields a total of $r\cdot n$ coded blocks, the ${\bf y}^{(i)}$ vectors and $n$ XOR parity blocks {\color{black}in} the ${\bf s}$ vector.
That is, we have an aggregate of $(r+1)n$ blocks available to place in $n$ nodes, hence we decide to store $r+1$ blocks per node.
{\color{black}Therefore, each node needs to have a storage capacity of}
 \begin{equation}
\alpha=\frac{M}{k}+\frac{1}{r}\frac{M}{k} =r+1 \text{ (coded blocks)}. \nonumber
\end{equation}

In Table \ref{table:LRC}, we state the circular placement of symbols in nodes of the first $(r+1)$-group .
\begin{table}
{\small
\begin{align}
\begin{array}{c}
 \\
\text{blocks of ${\bf y}^{(1)}$}\\
\vspace{0.14cm}\text{blocks of ${\bf y}^{(2)}$}\\
\vspace{0.14cm}\vdots\\
\text{blocks of ${\bf y}^{(r)}$}\\
\vspace{0.14cm}\text{blocks of ${\bf s}$}\vspace{-0.2cm}
\end{array}
\begin{array}{|c|c|c|c|c|}
\hline
\text{node }1 & \text{node }2 & \ldots& \text{node }r& \text{node }r+1\\
\hline
{\color{cyan}y^{(1)}_1} & {\color{cyan}y^{(1)}_2} & \ldots& {\color{cyan}y^{(1)}_{r}}& {\color{cyan}y^{(1)}_{r+1}}  \\
{\color{red}y^{(2)}_2} & {\color{red}y^{(2)}_3} & \ldots& {\color{black}y^{(2)}_{r+1}} & {\color{red}y^{(2)}_1}  \\
\vdots & \vdots & \vdots  & \vdots& \vdots\\
{\color{blue}y^{(r)}_r} & {\color{blue}y^{(r)}_{r+1}} & \ldots & {\color{blue}y^{(r)}_{r-2}} & {\color{blue}y^{(r)}_{r-1}}  \\
s_{r+1} & s_{1} & \ldots & s_{r-1} & s_{r} \\
\hline
\end{array}\nonumber
\end{align}
}
\caption{The first $r+1$ nodes in our code construction. These nodes belong to the first $(r+1)$-repair group. The nodes in the remaining repair groups have a block placement that follows the same circular-shifting pattern.}
\label{table:LRC}
\end{table}
There are three key properties of the block placement:
\begin{enumerate}
\item each node contains $r$ coded blocks coming from different ${\bf y}^{(l)}$ coded vectors and $1$ additional parity symbol,
\item  the blocks in the $r+1$ nodes of the $i$-th $(r+1)$-group have indices that appear only in that specific repair group, and
\item  the blocks of each row have indices that obey a circular pattern, {\it i.e.}, the first row of symbols has {\color{black}index ordering} $\{1,2,\ldots, r+1\}$, the second has ordering $\{2,3,\ldots, r+1,1\}$, and so on.
\end{enumerate}
In Fig.~\ref{fig:code}, we show an LRC of the above construction with $M=8$, $\alpha = 3$, $n=6$ and $k=4$, that has locality $2$.

\subsection{Repairing lost nodes}
\begin{figure}[h]
\hrulefill
\begin{center}
\scalebox{0.7}{
\begin{tikzpicture}
[node distance =0.01 cm and 1cm, every place/.style={rectangle,  thick,minimum size=6mm, minimum width=23.7mm}, bend angle=23]

\node[place, rounded corners=2pt,  minimum size=25mm, minimum width =30mm, dotted, draw=black, fill=magenta!50] (fn1) at (0,-0.4) {};

\node[place, rounded corners=2pt,  draw=black, fill=blue!20, label = above:{failed node 1}] (node11) {$y_1^{(1)}$};
\node[place, rounded corners=2pt,  draw=black, fill=red!40, below = of node11] (node12)  {$y_2^{(2)}$};
\node[place, rounded corners=2pt,  draw=black, fill=black!10, below = of node12] (node13)  {$s_3=y_3^{(1)}+y_3^{(2)}$};

\node[node distance =0.5 cm and 1cm, place, rounded corners=2pt,  draw=black, fill=blue!20, below = of node13, label = above:{node 2}] (node21) {$y_2^{(1)}$};
\node[place, rounded corners=2pt,  draw=black, fill=red!40, below = of node21] (node22)  {$y_3^{(2)}$};
\node[place, rounded corners=2pt,  draw=black, fill=black!10, below = of node22] (node23)  {$s_1 = y_1^{(1)}+y_1^{(2)}$};

\node[node distance =0.5 cm and 1cm, place, rounded corners=2pt,  draw=black, fill=blue!20, below = of node23, label = above:{node 3}] (node31) {$y_3^{(1)}$};
\node[place, rounded corners=2pt,  draw=black, fill=red!40, below = of node31] (node32)  {$y_1^{(2)}$};
\node[place, rounded corners=2pt,  draw=black, fill=black!10, below = of node32] (node33)  {$s_2 = y_2^{(1)}+y_2^{(2)}$};

\node[place, rounded corners=2pt,  draw=black, fill=blue!20, label = above:{newcomer node 1}] (newnode11) at (5,-3.6) {$y_1^{(1)}$};
\node[place, rounded corners=2pt,  draw=black, fill=red!40, below = of newnode11] (newnode12)  {$y_2^{(2)}$};
\node[place, rounded corners=2pt,  draw=black, fill=black!10, below = of newnode12] (newnode13)  {$s_3=y_3^{(1)}+y_3^{(2)}$};

\path (1.2,-3.7) edge[ semithick, post]  (3.8,-3.6);
\path (1.2,-5.5) edge[ semithick, post]  (3.8,-3.6);

\path (1.2,-6.1) edge[ semithick, post]  (3.8,-4.2);
\path (1.2,-2.4) edge[ semithick, post]  (3.8,-4.2);

\path (1.2,-3) edge[ semithick, post]  (3.8,-4.9);
\path (1.2,-4.8) edge[ semithick, post]  (3.8,-4.9);

\end{tikzpicture}}
\end{center}
\caption{We show an example of a failed node repair. 
The repair locality here is $2$ since $2$ remaining nodes are involved in reconstructing the lost information of the first node. 
Observe that we repair a failed node by simply transferring blocks: no block combinations are need to be performed at the sender nodes.
Once the blocks are transferred to a newcomer, a simple XOR suffices for reconstruction.}
\hrulefill
\label{fig:LRC_64_repair}
\end{figure}
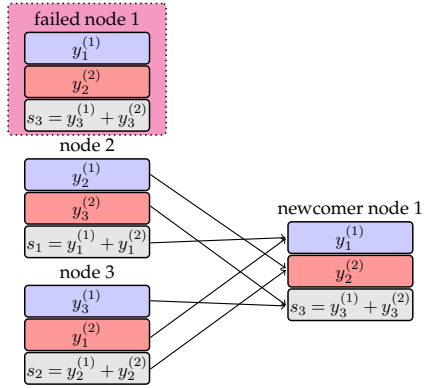

Here, we see that the repair of each lost node requires contacting $r$ nodes, {\it i.e.}, the locality of the code is $r$.
Without loss of generality, we consider the repair of a node in the first repair group of $r+1$ nodes.
This is sufficient since the nodes across different repair groups follow the same placement properties.

The key observation is that each node within a repair group stores $r+1$ blocks of {\it distinct} indices:
 the $r+1$ blocks of a particular index are stored in $r+1$ distinct nodes within {\color{black}a single} repair group.
When for example the first node fails, then $y_1^{(1)}$, the symbol of the first row, is regenerated by downloading $s_1$ from the second node, $y_1^{(r+1)}$ from the third, and so on.
Once all these symbols are downloaded, a simple XOR of all of them is exactly equal to $y_1^{(1)}$.
In the same manner, for each node, in each repair group when we need to reconstruct a lost block, we first download the $r$ remaining blocks of the same index and XOR them together to regenerate the desired lost block.
Since each block can be reconstructed by contacting $r$ other blocks, and since the repair is confined within {\color{black}a single repair group of $r$ remaining nodes}, the code has locality $r$.

In Fig.~\ref{fig:LRC_64_repair}, we show how repair is performed for the code construction presented in Fig.~\ref{fig:code}.

\subsection{Distance and code rate}
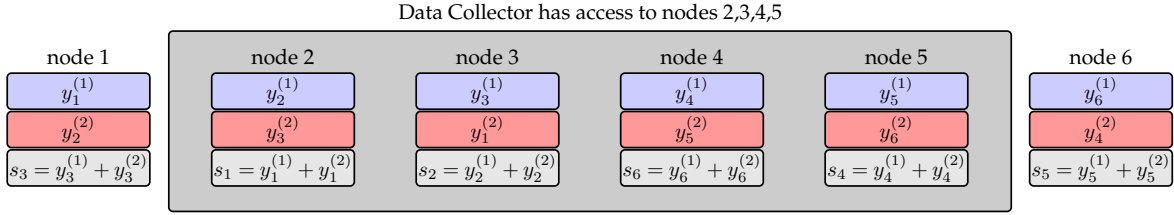
\begin{figure}[h]
\hrulefill
\begin{center}
\scalebox{0.8}{
\begin{tikzpicture}
[node distance =0.01 cm and 1cm, every place/.style={rectangle,  thick,minimum size=6mm, minimum width=23.7mm}, bend angle=23]

\node[place, rounded corners=2pt,  minimum size=30mm, minimum width =140mm, draw=black, fill=black!20, label = above:{Data Collector has access to nodes 2,3,4,5}] (group1) at (8.5,-0.5) {};
\node[place, rounded corners=2pt,  draw=black, fill=blue!20, label = above:{node 1}] (node11) {$y_1^{(1)}$};
\node[place, rounded corners=2pt,  draw=black, fill=red!40, below = of node11] (node12)  {$y_2^{(2)}$};
\node[place, rounded corners=2pt,  draw=black, fill=black!10, below = of node12] (node13)  {$s_3=y_3^{(1)}+y_3^{(2)}$};

\node[place, rounded corners=2pt,  draw=black, fill=blue!20, right = of node11, label = above:{node 2}] (node21) {$y_2^{(1)}$};
\node[place, rounded corners=2pt,  draw=black, fill=red!40, below = of node21] (node22)  {$y_3^{(2)}$};
\node[place, rounded corners=2pt,  draw=black, fill=black!10, below = of node22] (node23)  {$s_1 = y_1^{(1)}+y_1^{(2)}$};

\node[place, rounded corners=2pt,  draw=black, fill=blue!20, right = of node21, label = above:{node 3}] (node31) {$y_3^{(1)}$};
\node[place, rounded corners=2pt,  draw=black, fill=red!40, below = of node31] (node32)  {$y_1^{(2)}$};
\node[place, rounded corners=2pt,  draw=black, fill=black!10, below = of node32] (node33)  {$s_2 = y_2^{(1)}+y_2^{(2)}$};

\node[place, rounded corners=2pt,  draw=black, fill=blue!20, right = of node31, label = above:{node 4}] (node41) {$y_4^{(1)}$};
\node[place, rounded corners=2pt,  draw=black, fill=red!40, below = of node41] (node42)  {$y_5^{(2)}$};
\node[place, rounded corners=2pt,  draw=black, fill=black!10, below = of node42] (node43)  {$s_6=y_6^{(1)}+y_6^{(2)}$};

\node[place, rounded corners=2pt,  draw=black, fill=blue!20, right = of node41, label = above:{node 5}] (node51) {$y_5^{(1)}$};
\node[place, rounded corners=2pt,  draw=black, fill=red!40, below = of node51] (node52)  {$y_6^{(2)}$};
\node[place, rounded corners=2pt,  draw=black, fill=black!10, below = of node52] (node53)  {$s_4=y_4^{(1)}+y_4^{(2)}$};

\node[place, rounded corners=2pt,  draw=black, fill=blue!20, right = of node51, label = above:{node 6}] (node61) {$y_6^{(1)}$};
\node[place, rounded corners=2pt,  draw=black, fill=red!40, below = of node61] (node62)  {$y_4^{(2)}$};
\node[place, rounded corners=2pt,  draw=black, fill=black!10, below = of node62] (node63)  {$s_5=y_5^{(1)}+y_5^{(2)}$};

\end{tikzpicture}}
\end{center}

\caption{
We show how the file can be reconstructed by contacting $k=4$ nodes.
Observe that by accessing {\it any} $k$ nodes, a DC has access to $k$ blocks from the first MDS code and $k$ blocks from the second.
Since the pre-codes are $(n,k)$-MDS, this means that any $k$ blocks from each of the two coded blocks suffice to reconstruct both file parts.
}
\hrulefill
\label{fig:LRC_reconstruction}
\end{figure}

The distance of the presented code is $d=n-k+1$ due to the MDS pre-codes that are used in its design: any $k$ nodes in the system contain $rk$ distinct coded blocks, $k$ from each of the $r$ file blocks. Hence, by performing erasure decoding on each of these $r$ $k$-tuples of blocks, we can generate the $r$ blocks of the file.

When $(r+1)\nmid k$ this distance matches the optimal bound,
$$ n-\left\lceil \frac{M}{\alpha}\right\rceil-\left\lceil\frac{M}{r\alpha}\right\rceil+2 = n-\left\lceil \frac{kr}{r+1}\right\rceil-\left\lceil\frac{k}{r+1}\right\rceil+2=n-k+1,$$
since 
$$\lceil rk/(r+1) \rceil+\lceil k/(r+1) \rceil = k+\lceil- k/(r+1) \rceil+\lceil k/(r+1) \rceil=k+1,$$ when $k/(r+1)$ is not an integer \cite{knuth1989concrete}.
In Fig.~\ref{fig:LRC_reconstruction}, we give a file reconstruction example for the code of Fig. \ref{fig:code}.

Finally, the effective coding rate of  our LRC is
\begin{equation}
R = \frac{\text{size of useful information}}{\text{ total storage spent}}  =\frac{M}{n \cdot \alpha}= \frac{r}{r+1}\frac{k}{n}. \nonumber
\end{equation}
That is, the rate of the code is a fraction $\frac{r}{r+1}$ of the coding rate of an $(n,k)$ MDS code, hence is always upper bounded by
$\frac{r}{r+1}$.
 This loss in rate is incurred due to the use of the extra XOR stripe of blocks, that is required for efficient and local repairs.
Observe that if we set the repair locality to $r=f(k)$ and $f$ is a sub-linear function of $k$ ({\it i.e.}, $\log(k)$ or $\sqrt{k}$), then we obtain non-trivially low locality $r<<k$, while the excess storage cost $\epsilon=\frac{1}{r}$ is vanishing when $n,k$ grow.

\section{Conclusions}
In this work, we {\color{black}presented} {\it locally repairable codes}, a new family of repair efficient codes that optimize the metric of locality.
We analyze what is the best possible reliability in terms of code distance, given the requirement that each coded symbol can be reconstructed by $r$ other symbols in the code.
We provide an information theoretic bound that ties together the code distance, the locality, and the storage cost of a code.
We prove that this bound is achievable using vector-linear codes.
Eventually, we give an explicit construction of LRCs for the case where we require that any $k$ nodes can recover the encoded file.
We show how this explicit construction not only has optimal locality, but also requires small field size and admits very simple XOR based repairs.

%
%

\bibliographystyle{ieeetr}

\bibliography{LRC_IT}

\end{document}